\documentclass{article} 
\usepackage[utf8]{inputenc}
\usepackage{graphicx} 
\graphicspath{ {images/} }

\usepackage{geometry}
\usepackage{lscape}
\usepackage{rotating}

\usepackage{amsmath,amssymb,amsfonts,amsthm,float}
\usepackage{cancel}
\DeclareRobustCommand{\bbinom}{\genfrac{[}{]}{0pt}{}}
\usepackage{wrapfig}
\usepackage{subcaption}
\usepackage[nottoc]{tocbibind}
\usepackage[main=british,british]{babel}
\usepackage[backend=biber, style=numeric,]{biblatex}
\usepackage{csquotes}
\usepackage[compact]{titlesec}
\titleformat{\chapter}[display]
{\normalfont\huge\bfseries}{\chaptertitlename\ \thechapter}{20pt}{\Huge}

\usepackage[Lenny]{fncychap}
\ChNameVar{\fontsize{14}{16}
\usefont{OT1}{phv}{m}{n}\selectfont} 
\ChNumVar{\fontsize{60}{62}
\usefont{OT1}{ptm}{m}{n}\selectfont} \ChTitleVar{\Huge\bfseries\rm}
\ChRuleWidth{1pt}

\titlespacing*{\chapter}{0pt}{-40pt}{30pt}

\theoremstyle{plain}
\newtheorem{thm}{Theorem}[section]
\newtheorem{lem}[thm]{Lemma}
\newtheorem{prop}[thm]{Proposition}
\newtheorem{cor}[thm]{Corollary}

\theoremstyle{definition}
\newtheorem{defn}[thm]{Definition}

\newtheorem{exmp}[thm]{Example}

\theoremstyle{remark}

\newtheorem*{note}{Note}

\usepackage{chngcntr}
\counterwithin{table}{section}
\counterwithin{figure}{section}
\counterwithin{equation}{section}

\usepackage{mathrsfs}

\usepackage[table]{xcolor}
\definecolor{mycolor}{RGB}{104,36,109}
\definecolor{mycolor2}{RGB}{216, 172, 214}
\definecolor{mycolor3}{RGB}{230, 209, 229}
\definecolor{mycolor4}{RGB}{181, 233, 255}
\definecolor{headerSeahorse}{RGB}{183,184,226}
\definecolor{blockseahorse}{RGB}{214,214,238}
\definecolor{midseahorse}{RGB}{207,208,236}
\definecolor{mycolor5}{RGB}{254,232,255}
\definecolor{mycolor6}{RGB}{201,244,255}

\usepackage{chronosys}

 \catcode`\!=11
 
 \catcode`\!=12

\usepackage{adjustbox}

\usepackage{aligned-overset}
\usepackage{tabularx}
\usepackage{tocloft}
\usepackage{multirow}
\usepackage{longtable}

\usepackage{pgfplots}
\pgfplotsset{compat=1.15}
\usepackage{mathrsfs}
\usetikzlibrary{arrows}

\usetikzlibrary{arrows.meta}

\usepackage{fancyhdr}

\fancypagestyle{MyFancy}{%
    \fancyhf{}%
    \fancyhead[LE]{\fontsize{10}{12}\selectfont \rightmark}
    \fancyhead[RO]{\fontsize{10}{12}\selectfont \rightmark}
    \fancyfoot[C] {\thepage}%
    \setlength{\headheight}{14pt}
}

\usepackage{hyperref} 
\usepackage{setspace}

\begin{document}
\title{The MacWilliams Identity for Krawtchouk Association Schemes}
\author{Izzy Friedlander}
\date{\today}

\maketitle
\thispagestyle{plain}

\begin{abstract}

The MacWilliams Identity is a well established theorem relating the weight enumerator of a code to the weight enumerator of its dual. The ability to use a known weight enumerator to generate the weight enumerator of another through a simple transform proved highly effective and efficient. An equivalent relation was also developed by Delsarte \cite[Theorem 3]{DelsarteAlternating} which linked the eigenvalues of any association scheme to the eigenvalues of it's dual association scheme but this was less practical to use in reality. A functional transform was developed for some specific association schemes including those based on the rank metric, the skew rank metric and Hermitian matrices. In this paper those results are unified into a single consistent theory applied to these ``Krawtchouk association schemes" using a $b$-algebra. The derivatives formed using the $b$-algebra have also been applied to derive the moments of the weight distribution for any code within these association schemes.
\end{abstract}


\section{Introduction}

The idea of association schemes began to be developed through the 1940's. Early applications of this theory were in statistics and combinatorics as well as pure algebra, so this offered a neat way of cross fertilisation of methods in different areas of research. One benefit of symmetric association schemes was the ability to draw a graph that represented the points (vertices) and relations (edges) visually to explore their properties. It turns out there is a one to one relation between these graphs known as Distance Regular Graphs \cite{BrouwerDRG} and symmetric association schemes. Delsarte identified the potential for association schemes to advance the theory of codes \cite{delsarte1973algebraic} which is our primary interest here. 

The most well known association scheme is the Hamming scheme, based on vectors of length $n$ with the Hamming metric. This is the association scheme for which the MacWilliams Identity was first developed by MacWilliams herself \cite{TheoryofError}. The MacWilliams Identity takes the weight enumerator of a code written as a homogeneous polynomial and transforms it into that of another code (the dual code) through simple polynomial algebra. Using this identity a useful metric of the statistics of the code was found using basic calculus, the binomial moments of the weight enumerator. 

As the theory of codes advanced, new metrics were explored and, in turn, new association schemes. Delsarte pursued this line of thinking and applied the theory of the MacWilliams Identity to association schemes through the structure of the adjacency matrices of the association schemes themselves and their eigenvalues. This left the door open to explore ways in which the MacWilliams Identity as a functional transform could be extended to other schemes. 

In 2008 Gadouleau and Yan \cite{gadouleau2008macwilliams} successfully achieved this and presented a MacWilliams Identity as a functional transform of two homogeneous polynomials for codes based on the rank metric. This was exactly the same structure as that for the Hamming MacWilliams Identity, but the underlying algebra is significantly different. This proof applied character theory and the Hadamard transform to features of the matrices themselves. Once they had this theorem, they also derived the moments of the weight distribution of a code using the calculus of their $q$-algebra as well as investigating properties of maximal rank distance codes, again looking tantalisingly similar to those for the Hamming association scheme. 

In 2023, Friedlander et al. \cite{friedlander2023macwilliams} produced another functional transform MacWilliams Identity but now applied to the skew rank association scheme. Once more this comprised of two homogeneous polynomials and an appropriate underlying skew-$q$-algebra. The proof of this MacWilliams Identity was, however, different to the one by Gadouleau and Yan in that here the homogeneous polynomials used in the MacWilliams Identity were shown to generate the eigenvalues of the association scheme, so Delsarte's theory could then be invoked. 

This method was also then extended to the association scheme based on Hermitian matrices presented in \cite[Chapter 4]{IzzyThesis}. The nature of this scheme depended on parameters outside the range of validity specified by Delsarte for much of the previously used theory. So instead the results developed by Schmidt \cite{KaiHermitian} were heavily used to produce an equivalent MacWilliams Identity and resulting moments. 

Having looked at all four settings, a clear pattern emerged, but no obvious way of generalising them. The similarities between these four sets of results inspired the search for a unified theory that could be applied to these association schemes in general. Common elements included a specific homogeneous polynomial, a corresponding algebra and a corresponding recurrence relation that could be used to identify the eigenvalues. This analysis has led to a consolidated theory presented in this paper. Specific parameters are identified that recreate the individual results in the four cases previously studied. 

The rest of this paper is structured as follows: Section \ref{section:associationschemes} introduces association schemes and their parallel distance regular graphs. It describes the Bose-Mesner algebra of adjacency matrices and their eigenvalues, along with necessary definitions and properties and some important identities. Section \ref{section:generalPreliminaries} defines the Krawtchouk association schemes and the $b$-Krawtchouk polynomials which are then proven to be the eigenvalues of these schemes. Also included are the definitions of some required general weight functions. In Section \ref{section:thebalgebra} two integral components of the theory needed to produce a MacWilliams Identity, the $b$-algebra and the fundamental polynomials which together are used to generate the eigenvalues already mentioned, are introduced. The new generalised MacWilliams Identity as a functional transform is then presented and proved in Section \ref{section:generalisedMacWilliams}. Section \ref{section:bderivatives} follows an analogous approach to \cite[Section 5]{friedlander2023macwilliams} and \cite[Section 3]{gadouleau2008macwilliams} for finding the $b$-derivatives of this new algebra. The $b$-derivatives are then used in Section \ref{section:bmoments} to identify moments of the weight distribution of codes based on Krawtchouk association schemes. Finally Section \ref{section:individual} briefly reviews the parameters and results for each of the four individual association schemes and includes a clarification of the recurrence relation used in the Hermitian case. 

Having been able to extend the theory originally developed for the skew rank association scheme to Krawtchouk association schemes in general, the question remains; is it possible to extend this theory further to broader range of association schemes? For example, could the Eberlein polynomials \cite[Section 5.2]{delsartereccurance} be used to bring the Johnson and Grassmann schemes under this umbrella? Or could the work by Schmidt \cite{SymmetricKaiforms} be the key to encompass the Quadratic and Symmetric association schemes? 
\section{Association Schemes and Distance Regular Graphs}\label{section:associationschemes}

Association schemes give a particular structure to a set and that structure has been found to be useful when investigating properties of linear codes. Here we introduce the basic properties of an association scheme, focusing in particular on metric association schemes. We also identify their relationship with distance regular graphs which offers further understanding of these abstract concepts by visualisation. 

\subsection{Preliminaries}\label{subsection:preliminaries}
\begin{defn}\label{defn:associationscheme}
    A \textbf{\textit{symmetric association scheme with $\boldsymbol{n}$ classes}}, $\left(\mathscr{X},R\right)$, is defined as a finite set $\mathscr{X}$ of $v$ points and $n+1$ relations $R=\{R_0,\ldots, R_n\}$, which satisfy the following conditions
    \begin{align}
        R_0 & 
            = \left\{(x,x)~\vert~x\in\mathscr{X}\right\}
            \label{equation:relationsidentity}
            \\ 
        (x,y)\in R_i & \implies (y,x)\in R_i
            \label{equation:relationssymmetric}
            \\ 
        \left\{ R_0, R_1, \ldots, R_n\right\} 
            & \text{ is a partition of } \mathscr{X}\times\mathscr{X} 
            \label{equation:relationspartition}
            \\
            (x,y)\in R_k & \implies \left\vert \left\{ z\in\mathscr{X}~\vert~(x,z)\in R_i, (y,z)\in R_j\right\}\right\vert = c_{ijk} \label{equation:relationsintersectionnumbers}
    \end{align}
    where $c_{ijk}$ is a constant,  independent of the choice of $x$ and $y$ and are called the \textit{\textbf{intersection numbers}}. In other words, the relations satisfy having an identity, are symmetric, form a partition and have intersection numbers.
\end{defn}
We note that many different texts use the notation $p_{i,j}^{(k)}$ instead of $c_{ijk}$, as in \cite[(2.1)]{delsarte1973algebraic}.
\begin{defn}
    If $(x,y)\in R_i$ we call $x$ and $y$ the  \textbf{\textit{$\boldsymbol{i^{th}}$ associates}}. 
    
    The \textbf{\textit{valency}} \cite[p43]{BrouwerDRG} of each relation is defined as $v_i=c_{ii0}$, which is the number of $z\in\mathscr{X}$ such that $(x,z)\in R_i$. It is immediately obvious that $\sum_{i} c_{ii0}=v$.
\end{defn}
We note that there are non-symmetric association schemes but we are only focusing on those which are ``symmetric" in this paper. 

There are some well known identities for the valencies and the intersection numbers that can be useful when using the theory of association schemes \cite[Lemma 2.1.1]{BrouwerDRG}. The most interesting identity to note is that 
\begin{equation}
    \sum_{j=0}^{n} c_{ijk} = v_i.
\end{equation}
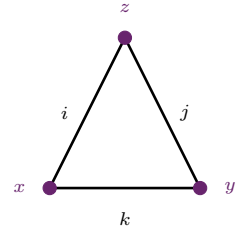
\begin{wrapfigure}{r}{0.4\textwidth}
    \centering
    \begin{tikzpicture}[line cap=round,line join=round,>=triangle 45,x=2cm,y=2cm]
\clip(0.2,0.2) rectangle (1.8,1.8);
\draw [line width=1pt] (0.5,0.5)-- (1,1.5);
\draw [line width=1pt] (1,1.5)-- (1.5,0.5);
\draw [line width=1pt] (1.5,0.5)-- (0.5,0.5);
\begin{scriptsize}
\draw [fill=mycolor,color=mycolor] (1,1.5) circle (2.5pt);
\draw[color=mycolor] (1,1.7) node {$z$};
\draw [fill=mycolor,color=mycolor] (0.5,0.5) circle (2.5pt);
\draw[color=mycolor] (0.3,0.5) node {$x$};
\draw [fill=mycolor,color=mycolor] (1.5,0.5) circle (2.5pt);
\draw[color=mycolor] (1.7,0.5) node {$y$};
\draw[color=black] (0.6,1) node {$i$};
\draw[color=black] (1.4,1) node {$j$};
\draw[color=black] (1,0.3) node {$k$};
\end{scriptsize}
\end{tikzpicture}
    \caption{Visualisation of points and relations in an association scheme.}
    \label{fig:triangle}    
\end{wrapfigure}
This identity can be explained in a bit more detail. Using Figure \ref{fig:triangle}, let $(x,y)\in R_k$. Then we see that
\begin{align}
    v_i & = \vert\{ z : (x,z)\in R_i\}|
    \\
        & = \sum_{j}\vert\{ z : (x,z)\in R_i, (z,y)\in R_j\}|
    \\
        & = \sum_{j} c_{ijk}. 
\end{align}

We go on to define the set of adjacency matrices that can be used to record and analyse the properties of the association scheme. 
\begin{defn}[{\cite[p613]{TheoryofError}}]
    The \textbf{\textit{adjacency matrix}}, ${D}_i$, of $R_i$ is defined to be a $v\times v$ matrix where each row and each column represents a point in $\mathscr{X}$ and where
    \begin{equation}\label{equation:adjacencymatrixcondition}
        \left( D_i \right)_{x,y} =
            \begin{cases}
                1 & \text{if } (x,y)\in R_i, \\
                0 & \text{otherwise.}
            \end{cases}
    \end{equation}
\end{defn}
\begin{lem}
Using the properties of an $(\mathscr{X},R)$ symmetric association scheme \cite[Lemma 2.1.1]{BrouwerDRG} we have that, 
\begin{align}
    D_0 & = I \label{equation:adjacencyidentity}
    \\
    D_i^{T} & = D_i
    \\
    \sum_{i=0}^{n}D_i & = 
    \begin{pmatrix}
        1 & \ldots & 1 \\
        \vdots & \ddots & \vdots \\
        1 & \ldots & 1
    \end{pmatrix} = J \label{equation:associtaionpartitions}
    \\
    D_iD_j & = \sum_{k=0}^n c_{ijk}D_k = D_jD_i, ~\text{for } i,j=0,\ldots,n. \label{equation:nastyadjacancyproperty}
    \\
    D_iJ & =JD_i = v_iJ \label{equation:assocationvi1s}
\end{align}
Conversely, any set of $\{0,1\}$ matrices, $D_0,\ldots,D_n$, that satisfies \eqref{equation:adjacencyidentity} - \eqref{equation:assocationvi1s} is the collection of adjacency matrices of an association scheme.
\end{lem}
\begin{proof}
    \begin{enumerate}
        \item[(1)] $D_0=I$ holds immediately from \eqref{equation:relationsidentity}.
        \item[(2)] $D_i^T=D_i$ holds immediately from \eqref{equation:relationssymmetric}. 
        \item[(3)] $\displaystyle\sum_{i=0}^{n}D_i=J$ holds immediately from \eqref{equation:relationspartition}. 
        \item [(4)] Proving Equation \eqref{equation:nastyadjacancyproperty}, if $(x,y)\in R_k$, then the matrix $(D_iD_j)_{x,y}=c_{ijk}$ by \eqref{equation:relationsintersectionnumbers}, and by \eqref{equation:adjacencymatrixcondition}, $(D_k)_{x,y}=1$. So we have \begin{equation}
            (D_iD_j)_{x,y} = c_{ijk}(D_k)_{x,y}.
        \end{equation}
        Now since $(D_iD_j)_{x,y}$ takes the intersection number of $(x,y)\in R_k$, we need to include all values of $k$, therefore
        \begin{equation}
            D_iD_j = \sum_{k=0}^{n}c_{ijk}D_k = D_jD_i
        \end{equation}
        as required.
        \item[(5)] Now we prove Equation \eqref{equation:assocationvi1s}. By definition of $D_i$, every row and column contains $v_i$ $1$'s. Therefore when we multiply by $J$, we sum these $1$'s, which is $v_i$ in every case by \eqref{equation:relationsintersectionnumbers}.
    \end{enumerate}
\end{proof}

\subsection{The Bose-Mesner Algebra}\label{section:bosemesner}

The eigenvalues of these adjacency matrices play in an important part in looking for optimal codes. That is, those with maximal distance for a given size. The algebra of these matrices was first explored by Bose and Nair \cite{BoseNair} in 1939 and later developed by Bose and Mesner \cite{BoseMesner} after whom it was named. The important results are outlined here.

To work with these adjacency matrices we define a set which consists of all complex linear combinations of the adjacency matrices. That is,
    \begin{equation}
        \mathscr{B} = \left\{ B = \sum_{i=0}^n b_i D_i ~\vert~ b_i \in \mathbb{C}\right\}.
    \end{equation}
This set then forms a ring with the operations of matrix addition and matrix multiplication with the added property of the multiplication being commutative by \eqref{equation:nastyadjacancyproperty}. As this set is a ring and also a vector space, it forms an algebra and is called the \textbf{\textit{Bose-Mesner Algebra}} of the association scheme.

Since $\mathscr{B}$ is commutative the members of $\mathscr{B}$ can be simultaneously diagonalised \cite{gantmacher2005}, i.e. there exists a single invertible matrix $P\in\mathscr{B}$, such that $P^{-1}BP$ is a diagonal matrix for each $B\in\mathscr{B}$. As a consequence there exists a unique alternative basis consisting of primitive idempodent matrices $E_0,\ldots,E_n$, of size $v\times v$. A primitive idempotent, $E_i$, is an idempotent such that it cannot be written as a direct sum of two other non-zero idempotents. So these idempotent matrices in the alternative basis satisfy the following equations,
\begin{align}
    E_i^2 & =E_i~\forall~{i} \label{equation:associationidempotent}\\
    E_iE_j & = 0 ~ \text{if } i \neq j \label{equation:associationij=0}\\
    \sum_{i=0}^n E_i & = I. 
\end{align}
It is conventional to choose $E_0=\frac{1}{v}J$. 

Since $\{E_0,\ldots,E_n\}$ is a basis for $\mathscr{B}$ there exist uniquely defined complex numbers $p_k(i)$ such that
\begin{equation}
    D_k = \sum_{i=0}^n p_k(i)E_i,\quad k=0,\ldots,n. \label{equation:associationDKsum}
\end{equation}

We also have, 
\begin{align}
    D_kE_i 
        & = \sum_{j=0}^{n} p_k(j) E_jE_i 
        \\
        \overset{\eqref{equation:associationij=0}}&{=} p_k(i) E_iE_i 
        \\
        \overset{\eqref{equation:associationidempotent}}&{=} p_k(i) E_i.\label{equation:eigenvaluesk=0=1}
\end{align}

Thus the $p_k(i)$'s are the \textbf{\textit{eigenvalues}} of $D_k$ by definition. The rank of each matrix $E_i$, denoted $\psi_i$, is the multiplicity of each eigenvalue $p_k(i)$ \cite[p45]{BrouwerDRG}. 

Since the $D_i$'s also form a basis of $\mathscr{B}$ we can also express each $E_k$ as a linear combination of the $D_i$. We then define
\begin{equation}
    E_k = \frac{1}{v} \sum_{i=0}^{n} q_k(i) D_i, \quad k=0,\ldots n
\end{equation}
such that $q_k(i)$ represent the coefficients of the change of basis matrix from $D_i$'s to the $E_i$'s. 

We then call the $q_k(i)$'s the \textbf{\textit{dual eigenvalues}} \cite{KreinParameters} of the association scheme. We also define the \textbf{\textit{eigenmatrices}} of the association scheme, $P=(p_{ik})$ and $Q=(q_{ik})$, to be the $(n+1)\times(n+1)$ matrices consisting of the eigenvalues $p_k(i)$ and $q_k(i)$ respectively. The eigenmatrices $P$ and $Q$ have the following properties.
\begin{thm}
    For $P$ and $Q$ eigenmatrices of a symmetric association scheme, we have,
    \begin{align}
        p_0(i) 
            & = q_0(i) = 1,
        \\
        p_k(0) 
            &= v_k, \quad q_k(0)=\psi_k 
        \\
        \sum_{i=0}^{n} \psi_kp_k(i)p_\ell(i) 
        & = v v_k\delta_{k\ell}\label{equation:associationorthog1}
        \\
        \sum_{i=0}^nv_iq_k(i)q_{\ell}(i) 
            & = v\psi_k\delta_{k\ell}\label{equation:associationorthog2}
        \\
        \psi_j p_i(j) 
            & = v_iq_j(i), \quad i,j=0,\ldots,n 
        \\
        \vert p_k(i) \vert 
            & \leq v_k, \quad \vert q_k(i) \vert \leq \psi_k 
            \\
        \sum_{j=0}^{n} p_k(i) & = \sum_{i=0}^n c_{iki}. 
    \end{align}
\end{thm}
    The proofs for these properties are well known and can be found at \cite[Lemma 2.2.1]{BrouwerDRG}. The equations \eqref{equation:associationorthog1} and \eqref{equation:associationorthog2} are called the \textbf{\textit{orthogonality relations}} \cite[Theorem 3]{TheoryofError}. For emphasis, $p_0(i)=1$ since by Equation \eqref{equation:eigenvaluesk=0=1} and we have $D_0=I$, immediately, $p_0(i)=1$.

We briefly introduce the idea of a \textit{\textbf{formal dual association scheme}}. Delsarte \cite[Section II C]{InfoTheoryDelsarte} proves that under some assumptions (explained below) on an $(\mathscr{X},R)$ $n$-class association scheme, you can find the dual association scheme which is an $(\mathscr{X}',R')$ $n$-class association scheme derived from the valencies, $v_i$, multiplicities, $\psi_i$ and eigenmatrices, $P$ and $Q$ of the original scheme by
\begin{equation}
    v'_i = \psi_i \quad \psi'_i = v_i \quad P' = Q \quad Q' = P.
\end{equation}
The ability to find a dual is, however, dependent on the scheme being ``regular'' as defined in \cite[Section 2.6.1]{delsarte1973algebraic}. We do not discuss regularity in more detail here as it is beyond the scope of this paper. In this paper Delsarte's condition of regularity is met and in fact, since $\mathscr{X}$ is always a vector space, the association schemes we consider here are all translation schemes. That is, if for all $z\in\mathscr{X}$, $i=0,\ldots,n$ we have $(x,y)\in R_i \implies (x+z,y+z)\in R_i$, then it is regular and we can find the scheme's dual.  Futhermore we only consider association schemes which are \textit{\textbf{formally self dual}}, i.e. when $P=Q$. In fact, the association schemes studied in this paper are all metric translation schemes. So when we discuss a code, $\mathscr{C}$ and its dual, $\mathscr{C}^\perp$ in an $(\mathscr{X},R)$ $n$-class association scheme, we have $\mathscr{C},\mathscr{C}^\perp\in\mathscr{X}$.

\subsection{Metric Association Schemes and Distance Regular Graphs}\label{subsection:metricassociationschemes}

In this paper we only consider association schemes that have a distance metric, which confers on the scheme an ordering of the relations. 

\begin{defn}\label{defn:metricscheme}
    For a $(\mathscr{X},R)$ symmetric $n$-class association scheme, we define a function $d:\mathscr{X}\times \mathscr{X}\rightarrow\mathbb{R}$ with $d(x,y)=i$ whenever $(x,y)\in R_i$ and the function satisfies the following conditions,
    \begin{align}
        d(x,y) & \geq 0
        \\
        d(x,y) & = 0  \iff x=y
        \\
        d(x,y) & = d(y,x).
    \end{align}
    We say the association scheme is \textbf{\textit{metric}} if $d$ is a metric on $\mathscr{X}$, i.e. for all $x,y,z\in \mathscr{X}$ we have
    \begin{equation}
        d(x,y)+d(y,z)\geq d(x,z).
    \end{equation}
    This is equivalent to for $(x,y)\in R_i$, $(y,z)\in R_j$, $(x,z)\in R_k$ we have that for $c_{ijk}\neq 0$, then $k\leq i+j$.
\end{defn}
If an association scheme is metric, then we can relate a graph $G$ with it by setting the edges to be $E=\{(x,y)\in R_1\}$ \cite[Chapter 1]{BrouwerDRG}. 


This graph is known as a distance regular graph and is defined below. For any $n$-class metric association scheme there is an associated distance regular graph, and vice versa, for every distance regular graph there is an associated $n$-class metric association scheme. 

\begin{defn}[{\cite[Chapter 1]{BrouwerDRG}}]
    A \textbf{\textit{distance regular graph}} is a graph $G=(V,E)$ in which for any two vertices $x,y\in V$, the number of vertices at distance $i$ from $x$ and $j$ from $y$ depend only on $i$ and $j$, and the distance between $x$ and $y$.
\end{defn}
Briefly, if we have a distance regular graph, then we define the points of the association scheme to be the vertices and $(x,y)\in R_1$ if there is an edge between $x$ and $y$. From that, we can say that $(x,z)\in R_i$ if the shortest path between $x$ and $z$ is length $i$.  

We now define a $P$-polynomial scheme and conclude that a metric association scheme is a $P$-polynomial scheme.
\begin{defn}[{\cite[p660]{TheoryofError}}]
    An association scheme is called a \textbf{\textit{P-polynomial}} scheme if there exists non-negative real numbers $z_0,\ldots,z_n$, with $z_0=0$ and real polynomials
    $\Phi_0(z),\ldots,\Phi_k(z)$, where  the degree of $\Phi_k(z)$, is $k$ such that
    \begin{equation}
        p_k(i)=\Phi_k(z_i), \quad i,k=0,\ldots,n.
    \end{equation}
\end{defn}
\begin{thm}[{\cite[Theorem 5.6, Theorem 5.16]{delsarte1973algebraic}}] An association scheme is metric if and only if it is a $P$-polynomial scheme, so the eigenvalues of the association scheme, $p_k(i)$, are indeed polynomials.
\end{thm}
Although not proven here, there are multiple ways of proving this statement. One is from MacWilliams and Sloane \cite[Theorem 6, Chapter 21]{TheoryofError}, another is from Brouwer \cite[Proposition 2.7.1]{BrouwerDRG} and originally proved by Delsarte \cite[Theorem 5.6, Theorem 5.16]{delsarte1973algebraic}.

Considering metric association schemes where $\mathscr{X}$ is a finite dimensional vector space over a finite field  and therefore a finite abelian group, we can introduce the concept of an \textit{\textbf{inner product}}, $\langle~,~\rangle$. More details on how the inner product arises in this situation can be found in \cite[p72]{BrouwerDRG}. 

In this situation, given an inner product, and since we have a finite vector space, we can identify a \textit{\textbf{dual vector subspace}}, $\mathscr{C}^\perp$, for any subspace $\mathscr{C}\subseteq\mathscr{X}$, such that
\begin{align}
    \mathscr{C}^\perp = \Big\{ x\in\mathscr{X}~|~ \langle x,y\rangle=0~ \forall ~y\in\mathscr{C}\Big\}.
\end{align}

In this paper we only consider finite dimensional vector spaces over a finite field, $\mathscr{X}$, so to find the orthogonal points we only need to consider when the inner product is $0$ and not involve character theory. We can note that as in Delsarte \cite[Section 3]{DelsarteBlinear} this would be equivalent to the character of the inner product being $1$ if the two points are orthogonal.

\subsection{Generalised Krawtchouk Polynomials}
From the orthogonality relations \eqref{equation:associationorthog1}, \eqref{equation:associationorthog2} we can see the polynomials, $p_k(i)$, form a set of polynomials which take the same form but with a range of parameters which we call a family. For the association schemes studied in this paper, this family has been shown to be the generalised Krawtchouk Polynomials introduced by Delsarte \cite{delsartereccurance}\cite{StantonKrawtchoukPolys}.

First we present a notation used by Delsarte \cite[p21]{DelsarteAlternating} called the $b$-nary Gaussian coefficients. Also note, for ease,  we define $\sigma_i=\frac{i(i-1)}{2}$ for $i\geq0$.

\begin{defn}
    For $x,k\in\mathbb{Z}^{+}$, $b\in\mathbb{R}$, $b\neq 1$ the \textbf{\textit{$b$-nary Gaussian coefficients}} are defined as
    \begin{equation}
        \prescript{}{b}{\bbinom{x}{k}} = \prod_{i=0}^{k-1}\frac{b^{x}-b^{i}}{b^{k}-b^{i}}
    \end{equation}
    with 
    \begin{equation}
        \prescript{}{b}{\bbinom{x}{0}} =1.
    \end{equation}
\end{defn}
    It is useful to note that if we take the limit as $b$ tends to 1, we in fact obtain the usual binomial coefficients,
    \begin{align}
        \lim_{b\to 1} \prod_{i=0}^{k-1} \prescript{}{b}{\frac{b^{x}-b^{i}}{b^{k}-b^{i}}}
            & = \lim_{b\to 1} \prod_{i=0}^{k-1}\frac{(b-1)}{(b-1)}\frac{\left(b^{x-i-1}+b^{x-i-2}+ \ldots +1\right)}{\left(b^{k-i-1}+ \ldots +1\right)}
            \\
            & = \prod_{i=0}^{k-1} \frac{x-i}{k-i}
            \\
            & = \binom{x}{k}.
    \end{align}
    This relationship helps when comparing the similarities between the analysis for the Hamming, the rank, the skew rank and the Hermitian association schemes.

Below are some identities relating to the $b$-nary Gaussian coefficients which are useful in simplifying notation, and can be used for different values of $b$ from \cite{DelsarteAlternating}. For $b\in\mathbb{R}/\{1\}$, $x,i,j,k\in\mathbb{Z}^+$, $y\in\mathbb{R}$ we have
\begin{align}
    \prescript{}{b}{\bbinom{x}{k}} 
        & = \prescript{}{b}{\bbinom{x}{x-k}}\label{equation:gaussianxx-k}
        \\
    \prescript{}{b}{\bbinom{x}{i}}\prescript{}{b}{\bbinom{x-i}{k}} 
        & = \prescript{}{b}{\bbinom{x}{k}}\prescript{}{b}{\bbinom{x-k}{i}}\label{equation:gaussianswapplaces}
            \\
    \prod_{i=0}^{x-1}\left(y-b^{i}\right)
        & = \sum_{k=0}^x (-1)^{x-k}b^{\sigma_{x-k}}\prescript{}{b}{\bbinom{x}{k}}y^k \label{equation:sumtoprod}
            \\
    \sum_{k=0}^x\prescript{}{b}{\bbinom{x}{k}}\prod_{i=0}^{k-1}\left(y-b^{i}\right)  
        & = y^x\label{equation:producttosumgauss}
            \\
    \sum_{k=i}^j (-1)^{k-i}b^{\sigma_{k-i}}\prescript{}{b}{\bbinom{k}{i}}\prescript{}{b}{\bbinom{j}{k}} 
        & = \delta_{ij}.\label{equation:deltaijbs}
\end{align}
The following identities are each used in the rest of this paper but can be shown trivially to be equal.
\begin{align}
    \prescript{}{b}{\bbinom{x}{k}} 
        & = \prescript{}{b}{\bbinom{x-1}{k}} + b^{x-k}\prescript{}{b}{\bbinom{x-1}{k-1}} \label{equation:Stepdown1}
        \\[5pt]
        & = \prescript{}{b}{\bbinom{x-1}{k-1}} + b^{k}\prescript{}{b}{\bbinom{x-1}{k}} \label{equation:gaussiancoeffsx-1k-1}
        \\[5pt]
        & = \frac{b^{x-k+1}-1}{b^{k}-1}\prescript{}{b}{\bbinom{x}{k-1}}
        \label{equation:gaussianfracxk-1}
        \\[5pt]
        & = \frac{b^{x}-1}{b^{x-k}-1}\prescript{}{b}{\bbinom{x-1}{k}} \label{equation:gaussianfracx-1k}
        \\[5pt]
        & = \frac{b^{x}-1}{b^{k}-1} \prescript{}{b}{\bbinom{x-1}{k-1}}. \label{equation:beta1stepdown}
        \end{align}
        
We also define a $b$-nary beta function which is closely related to the $b$-nary Gaussian coefficients, and aid us in notation throughout this paper.

\begin{defn}\label{defn:betafunction}
We define a \textbf{\textit{$b$-nary beta function}} for $x\in\mathbb{R}$, $k\in\mathbb{Z}^+$ as 
\begin{equation}\label{equation:betafunction}
    \beta_b(x,k) =
             \displaystyle\prod_{i=0}^{k-1}\prescript{}{b}{\bbinom{x-i}{1}}.
\end{equation}
\end{defn}

\begin{lem}\label{lemma:betabmanipulation}
We have for all $x\in\mathbb{R}$, $k\in\mathbb{Z}^+$,
    \begin{equation}\label{equation:betabstartdifferent}
        \beta_b(x,k) = \prescript{}{b}{\bbinom{x}{k}}\beta_b(k,k)
    \end{equation}
    \begin{equation}\label{equation:betabstartsame}
        \beta_b(x,x) = \prescript{}{b}{\bbinom{x}{k}}\beta_b(k,k)\beta_b(x-k,x-k)
    \end{equation}
and 
    \begin{equation}\label{equation:betaproperties}
        \beta_b(x,k)\beta_b(x-k,1)=\beta_b(x,k+1).
    \end{equation}
\end{lem}

\begin{proof}
We have
\begin{align}
    \beta_b(x,k) = \prod_{i=0}^{k-1}\prescript{}{b}{\bbinom{x-i}{1}} 
        & = \prod_{i=0}^{k-1}\frac{b^{x-i}-1}{b-1}
        \\
        & = \prod_{i=0}^{k-1}\frac{\left(b^{x-i}-1\right)\left(b^{k-i}-1\right)}{\left(b^{k-i}-1\right)(b-1)}
        \\
        & = \prod_{i=0}^{k-1}\frac{b^{x}-b^{i}}{b^{k}-b^{i}}\prod_{i=0}^{k-1}\frac{b^{k-i}-1}{b-1}
        \\
        & = {\prescript{}{b}{\bbinom{x}{k}}}\beta_b(k,k)
\end{align}
as required. Now we have
\begin{align}
    {\prescript{}{b}{\bbinom{x}{k}}}\beta_b(k,k)\beta_b(x-k,x-k) 
        & = \prod_{i=0}^{k-1}\frac{b^{x}-b^{i}}{b^{k}-b^{i}}\prod_{r=0}^{k-1}\frac{b^{k-r}-1}{b-1}\prod_{s=0}^{x-k-1}\frac{b^{x-k-s}-1}{b-1}
        \\
        & = \prod_{i=0}^{x-1}\frac{b^{x-i}-1}{b-1}
        \\
        & = \beta_b(x,x)
\end{align}
as required. And finally we have,
\begin{align}
        \beta_b(x,k)\beta_b(x-k,1) 
            & = \prescript{}{b}{\bbinom{x-k}{1}} \prod_{i=0}^{k-1}\prescript{}{b}{\bbinom{x-i}{1}}
            \\
            & = \beta_b(x,k+1).
    \end{align}
\end{proof}
To aid us in notation, we define a new $b$-nary gamma function, which is a component of the expression derived from setting $x=0$ in the generalised Krawtchouk polynomials \eqref{equation:delsarteKrawtchoukgenerlaised}.
\begin{defn}
    We define the \textit{\textbf{$b$-nary gamma function}} for $x,b,c\in\mathbb{R}$, $k\in\mathbb{Z}^+$, $cb>1$, for $c$ a constant, to be
    \begin{equation*}
    \gamma_{b,c}(x,k)=\displaystyle\prod_{i=0}^{k-1}\left(cb^x-b^{i}\right).
    \end{equation*}
\end{defn}

\begin{lem}\label{lemma:bGammaidentites}
We have the following identities for the $b$-nary Gamma function:
\begin{align}
    \gamma_{b,c}(x,k)
        & = b^{\sigma_k}\displaystyle\prod_{i=0}^{k-1}\left(cb^{x-i}-1\right)
    \\
    \gamma_{b,c}(x+1,k+1) 
        & =  \left(cb^{x+1}-1\right)b^{k}\gamma_{b,c}(x,k) \label{equation:gammastepdown}
    \\
    \gamma_{b,c}(x,k+1) 
        & = \left(cb^x-b^{k}\right)\gamma_{b,c}(x,k).\label{equation:gammastepdownsecond}
\end{align}
\end{lem}
    \begin{proof}~\\
$(1)$
\begin{align*}
   \gamma_{b,c}(x,k) 
        & = \prod_{i=0}^{k-1}\left(cb^x - b^{i}\right)
        \\
        & = \left(\prod_{i=0}^{k-1}b^{i}\right)\prod_{i=0}^{k-1}\left(cb^{x-i}-1\right)
        \\
        & = b^{\sigma_k}\prod_{i=0}^{k-1}\left(cb^{x-i}-1\right).
\end{align*}

$(2)$
\begin{align*}
        \gamma_{b,c}(x+1,k+1) 
                & =  \prod_{i=0}^{k} \left(cb^{x+1} - b^{i}\right)
                \\
                & =  \left(cb^{x+1}-1\right) \prod_{i=1}^{k}\left(cb^{x+1}-b^{i}\right)
                \\
                & =  \left(cb^{x+1}-1\right) \prod_{i=1}^{k} b\left(cb^{x}-b^{i-1}\right)
                \\
                & =  \left(cb^{x+1}-1\right) b^{k} \prod_{i=0}^{k-1} \left(cb^{x}-b^{i}\right)
                \\
                & = \left(cb^{x+1}-1\right)b^{k}\gamma_{b,c}(x,k).
\end{align*}
$(3)$
\begin{align*}
    \gamma_{b,c}(x,k+1) 
        & = \prod_{i=0}^k\left(cb^x-b^{i}\right)
        \\
        & = \left(cb^x-b^{k}\right)\prod_{i=0}^{k-1}\left(cb^x-b^{i}\right)
        \\
        & = \left(cb^x-b^{k}\right)\gamma_{b,c}(x,k).
\end{align*}
\end{proof}
\begin{note}
    The $b$-nary beta and $b$-nary gamma functions are new expressions which have been developed to unify the following theories in Hamming, rank, skew rank and Hermitian association schemes. 
\end{note}

Now that we have some additional notation, we can write Delsarte's \cite[(15)]{delsartereccurance} generalised Krawtchouk polynomials neatly and we shall write $\prescript{}{b}{\bbinom{x}{k}} ={\bbinom{x}{k}}$ for simplicity throughout the rest of this paper. 
\begin{defn}
    For $b,c\in\mathbb{R}$, $b\geq 1$, $c>\frac{1}{b}$, $n\in\mathbb{Z}^{+}$ $x,k\in\{0,\ldots,n\}$, then Delsarte's \cite{delsartereccurance} \textbf{\textit{generalised Krawtchouk polynomials}} are defined as
    \begin{equation}\label{equation:delsarteKrawtchoukgenerlaised}
        P_k(x,n) = \sum_{j=0}^{k}(-1)^{k-j}\left(cb^n\right)^j b^{\binom{k-j}{2}} {\bbinom{n-j}{n-k}}{\bbinom{n-x}{j}}.
    \end{equation}
\end{defn}

Again, if we take $b\rightarrow 1$, then the generalised Krawtchouk polynomials become the Krawtchouk polynomials in the usual sense. That is, the Hamming Krawtchouk polynomials for $q\geq 2$ are,
    \begin{equation}\label{equation:hammingKrawtchoukgenerlaised}
        P_k(x,n) = \sum_{j=0}^{k}(-1)^{k-j}q^j{\binom{n-j}{n-k}}{\binom{n-x}{j}}.
    \end{equation}
    We note that equation \eqref{equation:hammingKrawtchoukgenerlaised} is equal to \cite[(53),(55),(56)]{TheoryofError}, specifically for obtaining \cite[(56)]{TheoryofError} we use the substitution $j=``k-j"$ and rearrange the sum.

Delsarte proved that the eigenvalues of an association scheme satisfy a recurrence relation with specific initial values, namely for $b\in\mathbb{R}^+$, $y\in\mathbb{Z}^+$ and $x,k\in\{0,1,\ldots,y\}$ is 
\begin{equation}\label{equation:recurrencerelation}
    P_{k+1}(x+1,y+1) = b^{k+1}P_{k+1}(x,y)-b^kP_{k}(x,y)
\end{equation}
with initial values,
\begin{align}
    P_k(0,y) & = {\bbinom{y}{k}} \prod_{i=0}^{k-1} \left(cb^{y}-b^i\right)
    \\
    P_{0}(x,y) & = 1,
\end{align}
with $c\in\mathbb{R}$, $c>\frac{1}{b}$.
He then concluded that the only solution to this recurrence relation with these specific initial values are the generalised Krawtchouk polynomials as defined in Definition \ref{equation:delsarteKrawtchoukgenerlaised}. 

Delsarte also considers any association scheme to find a relationship between the inner distribution of an association scheme and its dual \cite[(6.9)]{delsarte1973algebraic}. Before we do this, we need to introduce some notation. 
\begin{defn}
    Let $(\mathscr{X},R)$ be an $n$-class association scheme. The \textbf{\textit{inner distribution}} of a subgroup $X\subseteq\mathscr{X}$, is the $(n+1)$-tuple, $\boldsymbol{c}=(c_0,\ldots,c_n)$, where $c_i$ is the average number of points of $X$ being $i^{th}$ associates of a fixed point of $X$.
\end{defn}
We note that in this paper we only consider association schemes where the inner distribution becomes a weight distribution with an associatied metric.            

\begin{thm}[The MacWilliams Identity for Association Schemes]\label{MacWilliamsDelsarte}
Let $(\mathscr{X},R)$ be an $n$-class association scheme with dual $n$-class association scheme $(\mathscr{X},R')$. For a pair of dual subgroups $X,X'\subseteq \mathscr{X}$, let $\boldsymbol{c}=(c_0,\ldots,c_n)$ be the inner distribution of $X$ and $\boldsymbol{c'}=(c'_0,\ldots c'_n)$ be the inner distribution of $X'$. If $P$ and $Q$ are the eigenmatrices of $(\mathscr{X},R)$ then
\begin{align}
    |X|\boldsymbol{c'} & = \boldsymbol{c}Q \\
    |X'|\boldsymbol{c} & = \boldsymbol{c'}P.
\end{align}
\end{thm}

\section{Krawtchouk Association Schemes}\label{section:generalPreliminaries}

Since we are considering the more general idea of metric schemes (and therefore $P$-polynomial schemes) we don't specify the details of any codes and spaces they are in, but rather consider parameters of their association schemes. In this paper we are only considering metric schemes which are self dual, i.e. when there is an ordering of the relations as in Definition \ref{defn:metricscheme} and the eigenmatrices $P$ and $Q$ coincide, and also only considering those with eigenvalues that satisfy Delsarte's recurrence relation \cite[(1)]{delsartereccurance} with specific initial values.

\subsection{Recurrence relation}

Below is the recurrence relation we will use to define our set of association schemes. The recurrence relation, for $b\in\mathbb{R}$, $b\neq0$ $n\in\mathbb{Z}^+$ and $x,k\in\{0,1,\ldots,n\}$ is 
\begin{equation}\label{equation:generalrecurrencerelation}
    F_{k+1}(x+1,n+1) = b^{k+1}F_{k+1}(x,n)-b^kF_{k}(x,n)
\end{equation}
for any function $F_k(x,n)$. It is noted that in using this recurrence we have slightly extended the ranges of the parameters of Delsarte's recurrence relation to be able to include the Hermitian association scheme (where $b=-q$). A proof that this recurrence relation is valid for this case is shown in Proposition \ref{prop:recurrencematch}. 

We can now define the set of the association schemes that we want to consider in this chapter.
\begin{defn}
    For an $(\mathscr{X},R)$ $n$-class formally self dual metric translation association scheme with defined parameters $b,c\in\mathbb{R}$ we say it is a \textbf{\textit{Krawtchouk association scheme}} if the eigenvalues, $P_k(x,n)$, for $x,k\in\{0,1,\ldots,n\}$ satisfy the recurrence relation \begin{equation}
    P_{k+1}(x+1,n+1) = b^{k+1}P_{k+1}(x,n)-b^kP_{k}(x,n)
    \end{equation}
    with specific initial values
    \begin{align}
    P_k(0,n) & = \bbinom{n}{k}\gamma_{b,c}(n,k)\label{equation:generalinitalconditions1}
    \\
    P_0(x,n) & = 1.\label{equation:generalinitalconditions2}
\end{align}
\end{defn}

In fact, we can find a new set of polynomials which satisfy the recurrence relation with these initial values and therefore are the eigenvalues of the Krawtchouk association schemes. 

\subsection{The \texorpdfstring{$b$}{b}-Krawtchouk Polynomials}

\begin{defn}
For an $n$-class Krawtchouk association scheme where $x,k\in\{0,1,\ldots,n\}$, $b\in\mathbb{R}$, $b\neq 0$, we define the \textbf{\textit{the $\boldsymbol{b}$-Krawtchouk Polynomial}} as
\begin{equation}
    C_k(x,n) = \sum_{j=0}^k (-1)^j b^{j(n-x)} b^{\sigma_j}\bbinom{x}{j}\bbinom{n-x}{k-j}\gamma_{b,c}(n-j,k-j).
\end{equation}
\end{defn}

The way these polynomials arise will be explained in Section \ref{section:thebalgebra}.
We first prove that the $C_k(x,n)$ satisfy the recurrence relation \eqref{equation:generalrecurrencerelation} and the initial values \eqref{equation:generalinitalconditions1}, \eqref{equation:generalinitalconditions2} and therefore are the eigenvalues of the Krawtchouk association schemes. 

\begin{prop}
For $b,c\in\mathbb{R}$, $b\neq 0$, $c>\frac{1}{b}$ and for all $x,k\in\{0,\ldots,n\}$ we have
    \begin{equation}
        C_{k+1}(x+1,n+1) = b^{k+1}C_{k+1}(x,n)-b^{k}C_k(x,n).\label{equation:generalrecurrenceCKI}
    \end{equation}
\end{prop}
\begin{proof}
We look at all three terms sequentially. First, noting that $\displaystyle\bbinom{x}{j-1}=0$ when $j=0$,
\begin{align}
    C_{k+1}&(x+1,n+1)  \\
        & = \sum_{j=0}^{k+1} (-1)^j b^{j(n-x)} b^{\sigma_{j}} \bbinom{x+1}{j}\bbinom{n-x}{k+1-j} \gamma\left(n+1-j, k+1-j\right)  \\
        & = \left. C_{k+1}(x+1,n+1)\right|_{j=k+1} 
         \\
        & \hspace{1cm} \overset{\eqref{equation:gaussiancoeffsx-1k-1}}{+} \sum_{j=0}^{k} (-1)^j b^{j(n-x)+\sigma_{j}} \left\{\bbinom{x}{j-1}+b^{j}\bbinom{x}{j}\right\}\bbinom{n-x}{k+1-j}\gamma\left(n+1-j, k+1-j\right) 
        \\
        & = \left. C_{k+1}(x+1,n+1)\right|_{j=k+1}
        \\
        & \hspace{1cm} + \sum_{j=1}^{k} (-1)^j b^{j(n-x)+\sigma_{j}} \bbinom{x}{j-1}\bbinom{n-x}{k+1-j}\gamma_{b,c}(n+1-j, k+1-j)\label{equation:balpha1}
        \\
        & \hspace{1cm}  \overset{\eqref{equation:gammastepdown}}{+} \sum_{j=0}^k (-1)^j cb^{j(n-x)+\sigma_{j}+n+1+k-j}\bbinom{x}{j}\bbinom{n-x}{k+1-j} \gamma\left(n-j,k-j\right)\label{equation:balpha2}
        \\
        & \hspace{1cm} - \sum_{j=0}^k (-1)^j b^{j(n-x)+\sigma_{j}+k}  \bbinom{x}{j}\bbinom{n-x}{k+1-j}\gamma_{b,c}(n-j,k-j)\label{equation:balpha3} 
        \\
        & = \left. C_{k+1}(x+1,n+1) \right|_{j=k+1} + \alpha_1 + \alpha_2 + \alpha_3 
\end{align}

where $\alpha_1$, $\alpha_2$, $\alpha_3$ represent summands \eqref{equation:balpha1}, \eqref{equation:balpha2}, \eqref{equation:balpha3} respectively and for notation, $|_{j=k+1}$ means ``the term when $j=k+1$''. 

Second,
\begin{align}
 b^{k+1}C_{k+1}(x,n) 
    & = \sum_{j=0}^{k+1}(-1)^j b^{k+1}b^{j(n-x)} b^{\sigma_{j}}\bbinom{x}{j}\bbinom{n-x}{k+1-j}\gamma_{b,c}(n-j,k+1-j) 
    \\
     & =  b^{k+1}\left. C_{k+1}(x,n)\right|_{j=k+1} 
     \\
     & \hspace{1cm} \overset{\eqref{equation:gammastepdownsecond}}{+} \sum_{j=0}^{k} (-1)^j cb^{j(n-x)+\sigma_{j}+n+1+k-j} \bbinom{x}{j}\bbinom{n-x}{k+1-j} \gamma_{b,c}(n-j,k-j) 
     \\
     & \hspace{1cm} - \sum_{j=0}^k (-1)^j b^{j(n-x)+\sigma_{j}+k+k-j+1} \bbinom{x}{j}\bbinom{n-x}{k+1-j}\gamma_{b,c}(n-j, k-j)\label{equation:bbeta1}
     \\
     & = b^{k+1}\left. C_{k+1}(x,n)\right|_{j=k+1} + \alpha_2 + \beta_1 . 
\end{align}
Where $\beta_1$ represents the summand \eqref{equation:bbeta1}. 
Third,
\begin{equation}
    \begin{split}
        b^{k}C_k(x,n) & =\sum_{j=0}^k (-1)^j b^{j(n-x)+\sigma_{j}+k} \bbinom{x}{j}\bbinom{n-x}{k-j} \gamma_{b,c}(n-j,k-j),\\
        & = \rho, ~\text{say}.
    \end{split}
\end{equation}
So let $C=C_{k+1}(x+1,n+1)-b^{k+1}C_{k+1}(x,n)+b^{k}C_{k}(x,n)$. We have, 
\begin{equation}
    C =  \alpha_1 + \alpha_3 - \beta_1 +\rho + \left. C_{k+1}(x+1,n+1)\right|_{j=k+1}-b^{k+1} \left. C_{k+1}(x,n)\right|_{j=k+1}.
\end{equation}
Consider $\alpha_3-\beta_1+\rho$. Then
\begin{align}
    \alpha_3-\beta_1 
        & = \sum_{j=0}^k(-1)^{j+1} b^{j(n-x)+\sigma_{j}+k} \bbinom{x}{j}\bbinom{n-x}{k+1-j}\gamma_{b,c}(n-j,k-j)\left( 1-b^{k-j+1}\right) 
        \\
        \overset{\eqref{equation:gaussianfracxk-1}}&{=} \sum_{j=0}^k (-1)^{j+1}b^{j(n-x)+\sigma_{j}+k}\left(1-b^{k-j+1}\right) \bbinom{x}{j} 
        \\
        & \hspace{1cm} \times \dfrac{b^{(n-x)-(k-j)}-1}{b^{k+1-j}-1}\bbinom{n-x}{k-j} \gamma_{b,c}(n-j,k-j)  \\
        & = \sum_{j=0}^k (-1)^j b^{(j+1)(n-x)+\sigma_{j+1}} \bbinom{x}{j}\bbinom{n-x}{k-j}\gamma_{b,c}(n-j, k-j)\label{equation:btau}\\
        & \hspace{1cm} - \sum_{j=0}^k (-1)^{j}b^{j(n-x)+\sigma_{j}+k} \bbinom{x}{j}\bbinom{n-x}{k-j}\gamma_{b,c}(n-j,k-j)  \\
        & = \tau-\rho,
\end{align}

where $\tau$ represents the summand in \eqref{equation:btau}. Thus,
\begin{equation}
    C  = \alpha_1 + \tau + \left. C_{k+1}(x+1,n+1)\right|_{j=k+1} - b^{k+1}\left. C_{k+1}(x,n)\right|_{j=k+1}.
\end{equation}

Now, 
\begin{align}
C_{k+1}&\left.(x+1,n+1)\right|_{j=k+1} 
- b^{k+1}  \left. C_{k+1}(x,n)\right|_{j=k+1}
    \\
    & = (-1)^{k+1}b^{(k+1)(n-x)}b^{\sigma_{k+1}}\left\{\bbinom{x+1}{k+1}-b^{k+1}\bbinom{x}{k+1}\right\}
    \\
    \overset{\eqref{equation:gaussiancoeffsx-1k-1}}&{=} (-1)^{k+1}b^{(k+1)(n-x)}b^{\sigma_{k+1}}\bbinom{x}{k}
    \\
    & = -\tau|_{j=k}
\end{align}

Now consider $\alpha_1.$
\begin{align}
    \alpha_1 & = \sum_{j=1}^k(-1)^{j} b^{j(n-x)+\sigma_{j}}\bbinom{x}{j-1}\bbinom{n-x}{k+1-j}\gamma_{b,c}(n+1-j, k+1-j)\\
    & = \sum_{j=0}^{k-1} (-1)^{j+1}b^{(j+1)(n-x)+\sigma_{j+1}}\bbinom{x}{j}\bbinom{n-x}{k-j}\gamma_{b,c}(n-j, k-j)\\
    & = -\tau+\tau|_{j=k}.
\end{align}
Thus $C=0$ and so the $C_k(x,n)$ satisfy the recurrence relation \eqref{equation:generalrecurrencerelation}.
\end{proof}

\begin{lem}\label{lemma:generalckiequalspki}
The $C_k(x,n)$ are the eigenvalues of the Krawtchouk association scheme. In other words,
\begin{equation}\label{equation:generalckiequalspki}
    C_k(x,n) = P_k(x,n).
\end{equation}
\end{lem}
\begin{proof}
The $C_k(x,n)$ satisfy the recurrence relation \eqref{equation:generalrecurrenceCKI} and the initial values of the $C_k(x,n)$ are 
\begin{align}
    C_k(0,n) & = \sum_{j=0}^k (-1)^j b^{jn}b^{\sigma_j}\bbinom{0}{j}\bbinom{n}{k-j}\gamma_{b,c}(n-j, k-j)\\
    & = \bbinom{n}{k}\gamma_{b,c}(n,k)
\end{align}
as $\displaystyle\bbinom{0}{j}=0$ unless $j=0$ and
\begin{align}
    C_0(x,n) 
        & = (-1)^0 b^{0} b^{0}\bbinom{x}{0}\bbinom{n-x}{0}\gamma_{b,c}(n,0)
    \\
        & = 1
\end{align}
as required.
\end{proof}

We note that this form for the eigenvalues is distinct from the three forms presented in \cite[Section 5.1]{delsartereccurance}. 
\begin{exmp}\label{example:differentforms}
    Consider the association scheme of skew-symmetric matrices ($b=q^2,$ $c=q^{-1}$) with $t=4$, then $n=2$ and $m=3$. We let the 3 forms presented in \cite[Section 5.1]{delsartereccurance}, starting with Equation $(15)$, be $P_k(x,n)$, $Q_k(x,n)$ and $R_k(x,n)$ in the order they appear in the paper. Then looking term by term we have the resulting Table \ref{tab:example1} for $k=1$ and $x=1$. 
    \begin{table}[H]
        \centering
        \renewcommand{\arraystretch}{1.5}
        $\begin{array}{c|c|c|c}
        \text{Eigenvalues} & j=0 & j=1 & \sum_{j=0}^1
        \\
        \hline
        C_1(1,2) & q^3-1 & -q^2 & q^3-q^2-1 \\
        P_1(1,2) & -q^2-1 & q^3 & q^3-q^2-1 \\
        Q_1(1,2) & q^3-q^2 & -1 & q^3-q^2-1 \\
        R_1(1,2) & \left(q^2+1\right)\left(q^3-1\right) & -q^5 & q^3-q^2-1 \\
    \end{array}$
        \caption{Components of the eigenvalues for $C_k(x,n)$ compared to others}
        \label{tab:example1}
    \end{table}
    We can clearly see in this example that the sum of the terms is the same, but the individual components cannot be equated on a term by term basis. 
\end{exmp}
    
\subsection{Weight Functions}

Given that we are only working with translation association schemes where the set of points $\mathscr{X}$ is a vector space, we can always attribute a weight function for that scheme since we will always have a distance between points and a $0$ element. Mathematically speaking, if we let  $(\mathscr{X},R)$ be an $n$-class translation scheme, we say that if $x,y$ are $n$ distance apart, then $(x,y)\in R_n$. Since $\mathscr{X}$ is a vector space, then $x-y,0\in\mathscr{X}$. Consequently, since $x,y$ are distance $n$ apart, then $(x-y,0)\in R_n$ also. 
\begin{defn}
    For an $(\mathscr{X},R)$ $n$-class Krawtchouk association scheme and $x\in\mathscr{X}$, we define the \textbf{\textit{scheme weight}} of $x$ to be $\omega$ if and only if $(x,0)\in R_\omega$.
\end{defn}

\begin{defn}
    For an $(\mathscr{X},R)$ $n$-class Krawtchouk association scheme, and for all $x\in\mathscr{X}$ of weight $\omega$, the \textbf{\textit{scheme weight function}} of $x$, denoted $f_S(x)$, is defined as the homogeneous polynomial 
    \begin{equation}
        f_S(x) = Y^{\omega}X^{n-\omega}.
    \end{equation}
    Now let $\mathscr{C}\subseteq\mathscr{X}$ be a code. Suppose there are $c_i$ codewords in $\mathscr{C}$ with weight $i$ for $0\leq i \leq n$. Then the \textbf{\textit{scheme weight enumerator}} of $\mathscr{C}$, denoted $W_{\mathscr{C}}^{S}(X,Y)$, is defined as,
    \begin{equation}
        W_{\mathscr{C}}^{S}(X,Y) = \sum_{\zeta\in\mathscr{C}}f_S(\zeta)=\sum_{i=0}^{n}c_iY^{i}X^{n-i}.
    \end{equation}
    The $(n+1)$-tuple, $\boldsymbol{c}=(c_0,\ldots,c_n)$ of coefficients of the weight enumerator is called the \textbf{\textit{scheme weight distribution}} of the code $\mathscr{C}$.
\end{defn}

We note that since we are only working with metric association schemes, we can always define the \textbf{\textit{minimum distance}} of a code $\mathscr{C}$. Denoted $d_S(\mathscr{C})$ or $d_S$, it is simply the minimum distance between all possible pairs of codewords in $\mathscr{C}$, dependent on the metric being used. 

If we were to look at each individual association scheme (such as the skew rank association scheme) we would be able to look at counting the number of elements of a particular weight in the overall space using a combinatorial approach. In contrast here we use the valencies of the association scheme to identify those values in general.
\begin{thm}
    For $b,c\in\mathbb{R}$, $b\in\mathbb{R}$, $b\neq1$, $cb>1$ the number elements, $x\in\mathscr{X}$ with weight $\omega$ in an $(\mathscr{X},R)$ $n$-class Krawtchouk association scheme is
    \begin{equation}\label{equation:generalnumberofelements}
        \xi_{n,\omega}=\bbinom{n}{\omega}\gamma_{b,c}(n,\omega).
    \end{equation}
\end{thm}
\begin{proof}
    The number of elements of weight $\omega$, is the $\omega^{th}$ valency of the Krawtchouk association scheme. Since the $\omega^{th}$ valency is the initial value $P_{\omega}(0,n)$, the statement is proved. 
\end{proof}
A direct consequence of this is the ability to find the scheme weight enumerator of $\mathscr{X}$, denoted $\Omega_{n}$, as
\begin{equation}\label{equation:generalomega}
    \Omega_n = \sum_{i=0}^{n} \xi_{n,i}Y^{i}X^{n-i}.
\end{equation}

\section{The \texorpdfstring{$b$}{b}-Algebra}\label{section:thebalgebra}

Having studied some individual association schemes, it was noticed that the underlying algebra was different in each case. Thus here we define a more general $b$-algebra that can be used for each corresponding Krawtchouk association scheme in a unified form. This helps us express the relations between the weight enumerator of a code and the weight enumerator of the code's dual. 

\subsection{The \texorpdfstring{$b$}{b}-Product, \texorpdfstring{$b$}{b}-Power and \texorpdfstring{$b$}{b}-Transform}

\begin{defn}\label{b-proddefn}
    Let
    \begin{align}
        a(X,Y;\lambda) 
            & = \sum_{i=0}^{r} a_i(\lambda) Y^{i} X^{r-i}
            \\
        b(X,Y;\lambda) 
            & = \sum_{i=0}^{s} b_i(\lambda) Y^{i} X^{s-i}
    \end{align}
    be two homogeneous polynomials in $X$ and $Y$ with coefficients $a_i(\lambda)$ and $b_i(\lambda)$ respectively, which are real functions of $\lambda$ and are $0$ unless otherwise specified. For example $b_i(\lambda)=0$ if $i\notin\left\{0,1,\ldots,s\right\}$. The \textbf{\textit{$b$-product}}, $\ast$, of $a(X,Y;\lambda)$ and $b(X,Y;\lambda)$ is defined as
    \begin{align}
        c(X,Y;\lambda) 
            & = a(X,Y;\lambda) \ast b(X,Y;\lambda)\label{equation:bproduct}
            \\
            & = \sum_{u=0}^{r+s} c_u(\lambda) Y^{u} X^{r+s-u} \label{equation:bproductexplicit}
    \end{align} 
    with
    \begin{equation}
        c_u(\lambda) = \sum_{i=0}^{u} b^{is} a_i(\lambda)b_{u-i}(\lambda -i).
    \end{equation}
\end{defn}
We note that as with the $q$-product in \cite[Lemma 1]{gadouleau2008macwilliams}, the $b$-product is not commutative or distributive in general. However, if $a(X,Y;\lambda)=a$ is a constant independent of $\lambda$, the following two property holds:
\begin{equation}
    a\ast b(X,Y;\lambda) = b(X,Y;\lambda) \ast a = ab(X,Y;\lambda).
\end{equation}
Separately if the degree of $a(X,Y;\lambda)$ and $c(X,Y;\lambda)$ are the same then,
\begin{align}
    \bigg(a(X,Y;\lambda) + c(X,Y;\lambda)\bigg)\ast b(X,Y;\lambda) & = a(X,Y;\lambda)\ast b(X,Y;\lambda) 
    \\
    & \hspace{1cm} + c(X,Y;\lambda)\ast b(X,Y;\lambda)
\end{align}
and 
\begin{align}
    b(X,Y;\lambda) \ast \bigg(a(X,Y;\lambda) + c(X,Y;\lambda)\bigg) 
        & = b(X,Y;\lambda) \ast a(X,Y;\lambda)\\
        & \hspace{1cm} + b(X,Y;\lambda) \ast c(X,Y;\lambda).
\end{align}
\begin{defn}
    For $a(X,Y;\lambda) = \displaystyle\sum_{i=0}^{r} a_{i}(\lambda) Y^{i} X^{r-i}$ the \textbf{\textit{$b$-power}} is defined by 
    \begin{align}
            a^{[0]}(X,Y;\lambda) & = 1 \\
            a^{[1]}(X,Y;\lambda) & = a(X,Y;\lambda) \\
            a^{[k]}(X,Y;\lambda) & = a(X,Y;\lambda) \ast a^{[k-1]}(X,Y;\lambda) \quad \text{for } k \geq 2.
    \end{align}
\end{defn}

\begin{defn}[{\cite[Definition 4]{gadouleau2008macwilliams}}]\label{defn:btransform}
    Let $a(X,Y;\lambda) = \displaystyle\sum_{u=0}^{r} a_{i}(\lambda) Y^{i} X^{r-i}$. We define the \textbf{\textit{$b$-transform}} to be the homogeneous polynomial
    \begin{equation}
        \overline{a}(X,Y;\lambda) = \sum_{i=0}^{r} a_{i}(\lambda) Y^{[i]} \ast X^{[r-i]}.
    \end{equation}
\end{defn}

\subsection{Fundamental Polynomials}

Again, having studied some of the individual association schemes, it was important to also generalise the homogeneous polynomials which are used in the statement and proof of the MacWilliams Identity as a functional transform. 
Let
\begin{equation}\label{equation:generalmu}
    \mu(X,Y;\lambda) = X + \left(cb^{\lambda}-1\right)Y.
\end{equation}
The $b$-powers of $\mu(X,Y;n)$ provide an explicit form for the weight enumerator of $(\mathscr{X},R)$, the Krawtchouk association scheme with $n$ classes.
\begin{thm}\label{theorem:bmu}
    If $\mu(X,Y;\lambda)$ is as defined above, then
    \begin{equation} \label{equation:bmuformula}
        \mu^{[k]}(X,Y;\lambda) = \sum_{u=0}^{k} \mu_{u}(\lambda,k) Y^{u} X^{k-u} \quad \text{for } k \geq 1,
    \end{equation}
    where
    \begin{equation}
        \mu_{u}(\lambda,k) =\bbinom{k}{u}\gamma_{b,c}(\lambda,u).
    \end{equation}
    Specifically, the weight enumerators for $(\mathscr{X},R)$, the $n$-class Krawtchouk association scheme, denoted by $\Omega_{n}$, is given by
    \begin{equation}
        \Omega_{n} = \mu^{[n]}(X,Y;n).
    \end{equation}
\end{thm}

\begin{proof}
    The proof follows the method of induction. Consider $k=1$, so 
    \begin{equation}
        \mu^{[1]}(X,Y;\lambda) = \mu(X,Y;\lambda) = X + \left(cb^{\lambda}-1\right)Y.
    \end{equation}
    Then
    \begin{align}
        \mu_{0}(\lambda,1) 
            & = 1 = \bbinom{1}{0}\gamma_{b,c}(\lambda,0) 
            \\
        \mu_{1}(\lambda,1) 
            & = \left(cb^{\lambda}-1\right) = \bbinom{1}{1}\gamma_{b,c}(\lambda,1).
    \end{align}
    So 
    \begin{equation}
        \mu_{u}(\lambda,1) = \bbinom{1}{u}\gamma_{b,c}(\lambda,u)
    \end{equation}
    and
    \begin{equation}
        \Omega_{1} = \sum_{i=0}^{1}\bbinom{1}{u} \gamma_{b,c}(\lambda,u)Y^{i}X^{1-i}=\mu^{[1]}(X,Y;1)
    \end{equation}
    as required for $k=1$. Now assume  the theorem is true for $k \geq 1$. Then 
    \begin{align}
        \mu^{[k+1]}(X,Y;\lambda) 
            & = \mu(X,Y;\lambda) \ast \mu^{[k]} (X,Y;\lambda)
            \\
            & = \left(X+\left(cb^\lambda -1\right)Y\right) \ast \left(\sum_{u=0}^{k} \bbinom{k}{u}\gamma_{b,c}(\lambda,u) Y^{u} X^{k-u} \right)
            \\
            & = \sum_{i=0}^{k+1} f_i(\lambda) Y^{i} X^{k+1-i}
    \end{align}
    where, 
    \begin{align}
        f_i(\lambda) & = \sum_{j=0}^{i} b^{jk} \mu_{j}(\lambda,1) \mu_{i-j}(\lambda-j,k)
            \\
            & = \mu_{0}(\lambda,1) \mu_{i}(\lambda,k) + b^{k} \mu_{1}(\lambda,1) \mu_{i-1}(\lambda-1,k)
            \\
            & = \bbinom{k}{i} \gamma_{b,c}(\lambda,i) + b^{k}\left( cb^{\lambda}-1\right) \bbinom{k}{i-1} \gamma_{b,c}(\lambda-1,i-1) 
            \\
            \overset{\eqref{equation:gammastepdown}\eqref{equation:gaussianfracx-1k}}&{=} \frac{b^{k-i+1}-1}{b^{k+1}-1} \bbinom{k+1}{i} \gamma_{b,c}(\lambda,i) \overset{\eqref{equation:gaussianfracx-1k}}{+} b^{k} \frac{b^{i}-1}{b^{k+1}-1} b^{1-i} \bbinom{k+1}{i} \gamma_{b,c}(\lambda,i) 
            \\
            & = \gamma_{b,c}(\lambda,i) \bbinom{k+1}{i}\left( \frac{b^{k-i+1}-1+b^{k-i+1}\left(b^{i}-1\right)}{b^{k+1}-1} \right)
            \\
            & = \gamma_{b,c}(\lambda,i) \bbinom{k+1}{i}
    \end{align}
    since for $i\geq1$ we only need to consider the first two coefficients as when $j\geq 2$ then $\mu_{j}(\lambda,1)=0$. So it is true for $k+1$. Therefore by induction the first part of the theorem is true. Now consider $\mu^{[n]}(X,Y;n)$, then clearly
    \begin{align}
        \mu^{[n]}(X,Y;n) & = \sum_{u=0}^{n}\bbinom{n}{u}\gamma_{b,c}(n,u)Y^{u}X^{n-u}
        \\
        \overset{\eqref{equation:generalnumberofelements}}&{=} \sum_{u=0}^{n}\xi_{n,u}Y^{u}X^{n-u} \overset{\eqref{equation:generalomega}}{=} \Omega_n
    \end{align}
    as required.
\end{proof}

Now for the other fundamental polynomial. Interestingly we let $\nu(X,Y;\lambda) = X - Y$, which is the exact same polynomial for all Krawtchouk association schemes previously studied. 

\begin{thm}\label{lemma:bnulemma}
For all $k\geq 1$,
    \begin{equation}\label{equation:bnuformula}
        \nu^{[k]}(X,Y;\lambda) = \sum_{u=0}^{k} \nu_u(\lambda,k) Y^{u} X^{k-u}
    \end{equation}
    where
    \begin{equation}
        \nu_u(\lambda,k) = (-1)^{u}b^{\sigma_{u}} \bbinom{k}{u}.
    \end{equation}
\end{thm}

\begin{proof}
    We perform induction on $k$. For $k=1$ we have
    \begin{equation}
        \nu^{[1]}(X,Y;\lambda) = \nu(X,Y;\lambda) = X - Y.
    \end{equation}
    Clearly we also have
    \begin{equation}
        (-1)^{0}b^{\sigma_0}\bbinom{1}{0}Y^{0}X^{1} + (-1)^{1}b^{\sigma_1}\bbinom{1}{1}Y^{1}X^{0} = X - Y
    \end{equation}
    as required.
    Now assume the theorem holds for $k\geq 1$.
    \begin{align}
        \nu^{[k+1]}(X,Y;\lambda) 
            & = \nu(X,Y;\lambda) \ast \nu^{[k]} (X,Y;\lambda)
            \\
            & = (X-Y) \ast \left(\sum_{u=0}^{k} (-1)^{u} b^{\sigma_u} \bbinom{k}{u} Y^{u} X^{k-u} \right)
            \\
            & = \sum_{i=0}^{k+1} g_i(\lambda) Y^{i} X^{k+1-i}
    \end{align}
    where
    \begin{align}
        g_i(\lambda) 
            & = \sum_{j=0}^{i} b^{jk} \nu_{j}(\lambda,1)\nu_{i-j}(\lambda-j,k) \\
            & = b^{0}(1)(-1)^{i}b^{\sigma_{i}}\bbinom{k}{i} + b^{k}(-1)(-1)^{i-1}b^{\sigma_{i-1}}\bbinom{k}{i-1}
            \\
        \overset{\eqref{equation:gaussianfracx-1k}}&{=} (-1)^{i}b^{\sigma_{i}}\frac{b^{k+1-i}-1}{b^{k+1}-1}\bbinom{k+1}{i} \overset{\eqref{equation:beta1stepdown}}{+} b^{k}(-1)^{i}b^{\sigma_{i-1}}\frac{b^{i}-1}{b^{k+1}-1}\bbinom{k+1}{i}
            \\
            & = (-1)^{i}b^{\sigma_i}\bbinom{k+1}{i}\left\{ \frac{b^{k+1-i}-1}{b^{k+1}-1} + b^{k}b^{1-i}\frac{b^{i}-1}{b^{k+1}-1}\right\} 
            \\
            & = (-1)^{i}\frac{b^{\sigma_{i}}}{b^{k+1}-1}\bbinom{k+1}{i}\left\{b^{k+1-i} - 1 + b^{k+1} - b^{k+1-i}\right\}
            \\
            & = (-1)^{i}b^{\sigma_{i}}\bbinom{k+1}{i}
    \end{align}
    since for $i\geq1$ we only need to consider the first two coefficients as when $j\geq 2$ then $\nu_{j}(\lambda,1)=0$, thus the statement holds.  
\end{proof}

\section{The Generalised MacWilliams Identity}\label{section:generalisedMacWilliams}
We can now begin to put the final pieces of the puzzle together to be able to state and prove the MacWilliams Identity for a $n$-class Krawtchouk association scheme. Since we have proven that the $C_k(x,n)$ are also eigenvalues of the Krawtchouk association scheme, we can then invoke Delsarte's MacWilliams Identity (Theorem \ref{MacWilliamsDelsarte}) in the proof of our functional transform version.

\subsection{Generalised MacWilliams Identity}

Now we can write a generalised MacWilliams Identity as a functional transform for an $(\mathscr{X},R)$ $n$-class Krawtchouk association scheme.
Let the weight enumerator of $\mathscr{C}\subseteq\mathscr{X}$ be,
\begin{equation}
    W_{\mathscr{C}}^{S}(X,Y)=\sum_{i=0}^n c_i Y^{i} X^{n-i}
\end{equation}
and of its dual, $\mathscr{C}^{\perp}\subseteq\mathscr{X}$ be
\begin{equation}
    W_{\mathscr{C}^\perp}^{S}(X,Y)=\sum_{i=0}^n c_i' Y^{i} X^{n-i}.
\end{equation}

\begin{thm}[The MacWilliams Identity for an $n$-class Krawtchouk Association Scheme]\label{ASSOCIATIONMacWilliams}
    Let $\mathscr{C}\subseteq  \mathscr{X}$ be an linear $[n,k,d_S]$-code, with weight distribution $\boldsymbol{c}=(c_0,\ldots,c_n)$ and $\mathscr{C}^\perp\subseteq  \mathscr{X}$ its dual code, with weight distribution $\boldsymbol{c'}=(c_0',\ldots,c_n')$.
    Then
    \begin{align}
        W_{\mathscr{C}^\perp}^{S}(X,Y)
            & =\frac{1}{\left| \mathscr{C}\right|}\overline{W}_{\mathscr{C}}^{S}\left( X +(cb^n-1)Y, X-Y\right)\label{equation:generalMacWilliamsFunctional}
            \\
            & = \frac{1}{|\mathscr{C}|}\sum_{i=0}^n c_i(X-Y)^{[i]}\ast\left(X+\left(cb^n-1\right)Y\right)^{[n-i]}.\label{equation:generalrewriteMacWills}
    \end{align}
\end{thm}
\begin{proof}
For $0\leq i\leq n$ we have
\begin{align}
    \left(X-Y\right)^{[i]} 
        & \ast\left(X+\left(cb^n-1\right)Y\right)^{[n-i]}
        \\
        & = \left(\nu^{[i]}(X,Y;n)\right) \ast \left(\mu^{[n-i]}(X,Y;n)\right)
        \\
        \overset{\eqref{equation:bmuformula}\eqref{equation:bnuformula}}&{=}\left(\sum_{u=0}^i (-1)^u b^{\sigma_{u}}\bbinom{i}{u}Y^{u}X^{i-u}\right)\ast\left(\sum_{j=0}^{n-i}\bbinom{n-i}{j}\gamma_{b,c}(n,j)Y^{j}X^{n-i-j}\right)
        \\
        \overset{\eqref{equation:bproduct}}&{=}
        \sum_{k=0}^t\left(\sum_{\ell=0}^{k} (-1)^\ell b^{\ell(n-x)} b^{\sigma_\ell}\bbinom{x}{\ell}\bbinom{n-x}{k-\ell}\gamma_{b,c}(n-\ell,k-\ell)\right)Y^{k}X^{n-k}
        \\
        & = \sum_{k=0}^n C_k(i,n)Y^{k}X^{n-k}.
\end{align}
So then we have
\begin{align}
    \dfrac{1}{\left\vert\mathscr{C} 
    \right\vert} \overline{W}^{S}_\mathscr{C}\left(X+\left(cb^n-1\right)Y, X-Y\right)
        & = \dfrac{1}{\left\vert\mathscr{C}\right\vert}\sum_{i=0}^n c_i\left(X-Y\right)^{[i]}\ast\left(X+\left(cb^n-1\right)Y\right)^{[n-i]}
        \\
        & = \dfrac{1}{\left\vert\mathscr{C}\right\vert}\sum_{i=0}^n c_i \sum_{k=0}^n C_k(i,n)Y^kX^{n-k}
        \\
        & =\sum_{k=0}^n\left(\dfrac{1}{\left\vert\mathscr{C}\right\vert}\sum_{i=0}^n c_i C_k(i,n)\right)Y^{k}X^{n-k}
        \\
        \overset{\eqref{MacWilliamsDelsarte}}&{=} \sum_{k=0}^n c_k' Y^{k}X^{n-k}
        \\
        & = W_{\mathscr{C}^\perp}^{S}(X,Y).
\end{align}
\end{proof}
\section{The \texorpdfstring{$b$}{b}-Derivatives}\label{section:bderivatives}

In this section we develop a derivative. It should be clear that this section essentially reworks \cite[Section 5]{friedlander2023macwilliams} with a more general $b$-algebra. 

\subsection{The \texorpdfstring{$b$}{b}-Derivative}
To begin with, we consider the derivative with respect to $X$.  
\begin{defn}
For $b\neq 1$, the \textbf{\textit{$b$-derivative}} at $X\neq 0$ for a real-valued function $f(X)$ is defined as
\begin{equation}
    f^{(1)}\left(X\right)=\dfrac{f\left(bX\right)-f\left(X\right)}{(b-1)X}. 
\end{equation}
For $\varphi\geq0$ we denote the $\varphi^{th}$ $b$-derivative (with respect to $X$) of $f(X,Y;\lambda)$ as $f^{(\varphi)}(X,Y;\lambda)$. The $0^{th}$ $b$-derivative of $f(X,Y;\lambda)$ is $f(X,Y;\lambda)$. For any $a\in\mathbb{R},~X\neq0,$ and real-valued function $g(X)$,
\begin{equation}
    \Big[ f(X)+ag(X)\Big]^{(1)} = f^{(1)}(X)+ag^{(1)}(X).
\end{equation}
\end{defn}
Once again for the Hamming metric we take the formal definition of a derivative and take the limit of the function as $b\to 1$. That is, let $b=1+h$, $h\in\mathbb{R}$, then the derivative becomes,
\begin{equation}
    f^{(1)}(X) = \lim_{h\to 0}\frac{f((1+h)X)-f(X)}{hX}
\end{equation}
and so converts into the derivative in the usual sense of polynomials \cite[Problems (5), p98]{TheoryofError}. 

Now we have the definition of a derivative we can demonstrate some important results for homogeneous polynomials in general and the fundamental polynomials in particular. 


\begin{lem}\label{lemma:bderiv}
~\\
\begin{enumerate}
    \item For $0\leq \varphi \leq \ell,$ $\varphi\in\mathbb{Z}^{+}$ and $\ell \geq 0$,
        \begin{equation}
        \left(X^\ell\right)^{(\varphi)} = \beta_{b}(\ell,\varphi)X^{\ell-\varphi}.
        \end{equation}
    \item The $\varphi^{th}$ $b$-derivative of     $f(X,Y;\lambda)=\displaystyle\sum_{i=0}^r f_i(\lambda) Y^{i}X^{r-i}$ is given by
        \begin{equation}\label{equation:vthgeneralbderivative}
            f^{(\varphi)}\left(X,Y;\lambda\right)= \displaystyle\sum_{i=0}^{r-\varphi}f_i(\lambda) \beta_{b}(r-i,\varphi)Y^{i}X^{r-i-\varphi}.
        \end{equation}
    \item Also,
        \begin{align}
        \mu^{[k](\varphi)}(X,Y;\lambda) & = \beta_{b}(k,\varphi)\mu^{[k-\varphi]}(X,Y;\lambda)\label{equation:bmuderiv}
        \\
        \nu^{[k](\varphi)}(X,Y;\lambda) & = \beta_{b}(k,\varphi)\nu^{[k-\varphi]}(X,Y;\lambda).\label{equation:bnuderiv}
\end{align}
\end{enumerate}
\end{lem}

\begin{proof}~\\
\begin{enumerate}
    \item[(1)] For $\varphi=1$ we have
        \begin{equation}
            \left(X^{\ell}\right)^{(1)} = \dfrac{\left(bX^{}\right)^\ell-X^{\ell}}{(b-1)X} = \dfrac{b^{\ell}-1}{b-1}X^{\ell-1} = \bbinom{\ell}{1}X^{\ell-1} = \beta_{b}(\ell,\varphi)X^{\ell-1}.
        \end{equation}
The rest of the proof follows by induction on $\varphi$ and is omitted. 
    \item[(2)] Now consider $f(X,Y;\lambda)=\displaystyle\sum_{i=0}^r f_i (\lambda)Y^{i}X^{r-i}$. We have,
        \begin{align}
            f^{(1)}\left(X,Y;\lambda\right) 
            & = \left(\displaystyle\sum_{i=0}^r f_i (\lambda) Y^{i}X^{r-i}\right)^{(1)}
            \\
            & =\displaystyle\sum_{i=0}^r f_i(\lambda) Y^{i}\left(X^{r-i}\right)^{(1)}
            \\
            & =\displaystyle\sum_{i=0}^{r-1}f_i(\lambda) \beta_{b}(r-i,\varphi)Y^{i}X^{r-i-1}.
\end{align}

So the initial case holds. The rest of the proof follows by induction on $\varphi$ and is omitted. 
    \item[(3)] Now consider $\mu^{[k]}(X,Y;\lambda)=\displaystyle\sum_{u=0}^k \mu_u(\lambda,k)Y^{u}X^{k-u}$ where $\mu_u(\lambda,k) = \displaystyle\bbinom{k}{u}\gamma_{b,c}(\lambda,u)$ as in Theorem \ref{theorem:bmu}. Then we have
    \begin{align}
        \mu^{[k](1)}(X,Y;\lambda)
            & = \left(\sum_{u=0}^k \mu_u(\lambda,k)Y^{u}X^{k-u}\right)^{(1)}
            \\
            & = \sum_{u=0}^k \mu_u(\lambda,k)Y^{u}\left(\frac{\left(bX\right)^{k-u}-X^{k-u}}{(b-1)X}\right)
            \\
            & = \sum_{u=0}^{k-1}\frac{b^{k-u}-1}{b-1}\bbinom{k}{u}\gamma_{b,c}(\lambda,u)Y^{u}X^{k-u-1}
            \\
            \overset{\eqref{equation:gaussianfracx-1k}}&{=}  \sum_{u=0}^{k-1}\frac{(b^{k}-1)\left(b^{k-u}-1\right)}{(b^{k-u}-1)(b-1)}\bbinom{k-1}{u}\gamma_{b,c}(\lambda,u)Y^{u}X^{k-u-1}
            \\
            & = \left(\frac{b^{k}-1}{b-1}\right)\mu^{[k-1]}(X,Y;\lambda)
            \\
            \overset{\eqref{equation:betafunction}}&{=}\beta_{b}(k,1)\mu^{[k-1]}(X,Y;\lambda)
    \end{align}

So $\mu^{[k](\varphi)}(X,Y;\lambda)=\beta_{b}(k,\varphi)\mu^{[k-\varphi]}(X,Y;\lambda)$ follows by induction on $\varphi$ and is omitted.

Now consider $\nu^{[k]}(X,Y;\lambda)=\displaystyle\sum_{u=0}^k (-1)^u b^{\sigma_{u}}\displaystyle\bbinom{k}{u}Y^{u}X^{k-u}$ as in Theorem \ref{lemma:bnulemma}. Then we have
    \begin{align}
        \nu^{[k](1)}(X,Y;\lambda) 
            & = \sum_{u=0}^k (-1)^u b^{\sigma_{u}}\frac{b^{k-u}-1}{b-1}\bbinom{k}{u} Y^{u}X^{k-u-1}
            \\
            \overset{\eqref{equation:gaussianfracx-1k}}&{=} \sum_{u=0}^{k-1}(-1)^u b^{\sigma_{u}}\frac{\left(b^{k}-1\right)\left(b^{k-u}-1\right)}{\left(b^{k-u}-1\right)(b-1)}\bbinom{k-1}{u}Y^{u}X^{k-1-u}
            \\
            & = \frac{b^{k}-1}{b-1}\nu^{[k-1]}(X,Y;\lambda)
            \\
            \overset{\eqref{equation:betafunction}}&{=} \beta_{b}(k,1)\nu^{[k-1]}(X,Y;\lambda).
    \end{align}
So the initial case holds. Thus $\nu^{[k](\varphi)}(X,Y;\lambda)=\beta_{b}(k,\varphi)\nu^{[k-\varphi]}(X,Y;\lambda)$ follows by induction also and is omitted. 
\end{enumerate}
\end{proof}

We now need a few smaller lemmas in order to prove the Leibniz rule for the $b$-derivative.

\begin{lem}\label{lemma:blittleuandv}
Let
    \begin{align}
        u\left(X,Y;\lambda\right) 
            & = \sum_{i=0}^r u_i(\lambda) Y^{i}X^{r-i}
            \\
        v\left(X,Y;\lambda\right) 
            & = \sum_{i=0}^s v_i(\lambda) Y^{i}X^{s-i}.
    \end{align}
\begin{enumerate}
    \item If $u_r(\lambda)=0$ then
        \begin{equation}\label{equation:burequals0}
            \frac{1}{X}\Big[u\left(X,Y;\lambda\right)\ast v\left(X,Y;\lambda\right)\Big] = \frac{u\left(X,Y;\lambda\right)}{X}\ast v\left(X,Y;\lambda\right).
        \end{equation}
    \item If $v_s(\lambda)=0$ then
        \begin{equation}\label{equation:bvsequals0}
            \frac{1}{X}\Big[u\left(X,Y;\lambda\right)\ast v\left(X,Y;\lambda\right)\Big] = u\left(X, bY;\lambda\right)\ast \frac{v\left(X,Y;\lambda\right)}{X}.
        \end{equation}
\end{enumerate}
\end{lem}
\begin{proof}
    \begin{enumerate}
        \item[(1)] If $u_r(\lambda)=0$,
            \begin{equation}
                \frac{u\left(X,Y;\lambda\right)}{X} = \sum_{i=0}^{r-1}u_i(\lambda) Y^{i}X^{r-i-1}.
            \end{equation}
            Hence
            \begin{align}
                \frac{u\left(X,Y;\lambda\right)}{X}\ast v\left(X,Y;\lambda\right)
                    & = \sum_{k=0}^{r+s-1}\left(\sum_{\ell=0}^k b^{\ell s} u_\ell(\lambda) v_{k-\ell}(\lambda-\ell)\right) Y^{k}X^{r+s-1-k}
                    \\
                    & = \frac{1}{X}\sum_{k=0}^{r+s-1}\left(\sum_{\ell=0}^k b^{\ell s} u_\ell(\lambda) v_{k-\ell}(\lambda-\ell)\right) Y^{k}X^{r+s-k}
                    \\
                    & \hspace{1cm} + \frac{1}{X}\sum_{\ell=0}^{r+s} b^{\ell s} u_\ell(\lambda) v_{r+s-\ell}(\lambda-\ell) Y^{r+s}X^{0}
                    \\
                    & = \frac{1}{X}\left(u\left(X,Y;\lambda\right)\ast v\left(X,Y;\lambda\right)\right)
            \end{align}
        since $v_{r+s-\ell}(\lambda-\ell)=0$ for $0\leq \ell \leq r-1$ and $u_{\ell}(\lambda)=0$ for $r\leq \ell \leq r+s$. So 
        \begin{equation}
            \frac{1}{X}\sum_{\ell=0}^{r+s} b^{\ell s} u_\ell(\lambda) v_{r+s-\ell}(\lambda-\ell) Y^{r+s}X^{0} = 0.
        \end{equation}
    \item[(2)] Now if $v_s(\lambda)=0$,
        \begin{equation}
            \frac{v\left(X,Y;\lambda\right)}{X} = \sum_{i=0}^{s-1} v_i(\lambda) Y^{i}X^{s-1-i}.
        \end{equation}
        Then
        \begin{align}
            u\left(X,bY;\lambda\right) \ast \frac{v\left(X,Y;\lambda\right)}{X} 
                & = \sum_{k=0}^{r+s-1}\left(\sum_{\ell=0}^k b^{\ell(s-1)} b^{\ell}u_\ell(\lambda) v_{k-\ell}(\lambda-\ell)\right)Y^{k}X^{r+s-1-k}
                \\
                & = \sum_{k=0}^{r+s-1}\left(\sum_{\ell=0}^k b^{\ell(s-1)} b^{\ell}u_\ell(\lambda) v_{k-\ell}(\lambda-\ell)\right)Y^{k}X^{r+s-1-k}
                \\
                & \hspace{1cm} + \frac{1}{X}\sum_{\ell=0}^{r+s} b^{\ell s} u_\ell(\lambda) v_{r+s-\ell}(\lambda-\ell) Y^{r+s}X^{0}
                \\
                & = \frac{1}{X}\left[u(X,Y;\lambda)\ast v(X,Y;\lambda)\right]
        \end{align}
        \end{enumerate}
        \hspace{1cm} since $v_{r+s-\ell}(\lambda-\ell)=0$ for $0\leq \ell \leq r$ and $u_{\ell}=0$ for $r+1\leq \ell \leq r+s$.
\end{proof}

\begin{thm}[Leibniz rule for the $b$-derivative] \label{theorem:bLeibniz}
For two homogeneous polynomials in $X$ and $Y$, $f(X,Y;\lambda)$ and $g(X,Y;\lambda)$ with degrees $r$ and $s$ respectively, the $\varphi^{th}$ (for $\varphi\geq0$) $b$-derivative of their $b$-product is given by
    \begin{equation}\label{equation:bleibniz}
        \Big[ f\left(X,Y;\lambda\right)\ast g\left(X,Y;\lambda\right)\Big]^{(\varphi)} = \sum_{\ell=0}^\varphi\bbinom{\varphi}{\ell} b^{(\varphi-\ell)(r-\ell)}f^{(\ell)}\left(X,Y;\lambda\right)\ast g^{(\varphi-\ell)}\left(X,Y;\lambda\right).
    \end{equation}
\end{thm}

\begin{proof}
For simplification, we shall write $f(X,Y;\lambda)$ as $f(X,Y)$ and similarly $g(X,Y;\lambda)$ as $g(X,Y)$. Now by differentiation we have
\begin{align}
    \Big[ f\left(X,Y\right)\ast g\left(X,Y\right)\Big]^{(1)} 
        & = \frac{f\left(bX,Y\right)\ast g\left(bX,Y\right)-f\left(X,Y\right)\ast g\left(X,Y\right)}{(b-1)X}
        \\
        & = \frac{1}{(b-1)X} \bigg\{ f\left(bX,Y\right)\ast g\left(bX,Y\right)-f\left(bX,Y\right)\ast g\left(X,Y\right)
        \\
        & \hspace{1cm} + f\left(bX,Y\right)\ast g\left(X,Y\right) - f\left(X,Y\right)\ast g\left(X,Y\right)\bigg\}
        \\
        & = \frac{1}{(b-1)X}\Big\{ f\left(bX,Y\right)\ast\left(g\left(bX,Y\right)-g\left(X,Y\right)\right)\Big\}
        \\
        & \hspace{1cm} +\frac{1}{(b-1)X}\bigg\{\left(f\left(bX,Y\right)-f\left(X,Y\right)\right)\ast g\left(X,Y\right)\bigg\}
        \\
        \overset{\eqref{equation:bvsequals0}}&{=} f\left(bX,bY\right) \ast \left\{\frac{g\left(bX,Y\right)-g\left(X,Y\right)}{(b-1)X}\right\}
        \\
        &\hspace{1cm} \overset{\eqref{equation:burequals0}}{+}\left\{\frac{f\left(bX,Y\right)-f\left(X,Y\right)}{(b-1)X}\right\}\ast g\left(X,Y\right)
        \\
        & = b^{r} f\left(X,Y\right)\ast g^{(1)}\left(X,Y\right) + f^{(1)}\left(X,Y\right)\ast g\left(X,Y\right)
\end{align}

since $g(X,Y)$ has the same degree of $g(bX,Y)$ and similarly, $f(X,Y)$ has the same degree as $f(bX,Y)$. So the initial case holds. Assume the statement holds true for $\varphi=\overline{\varphi}$, i.e.

\begin{equation}
    \Big[ f\left(X,Y\right)\ast g\left(X,Y\right)\Big]^{(\overline{\varphi})} = \sum_{\ell=0}^{\overline{\varphi}} \bbinom{\overline{\varphi}}{\ell}b^{(\overline{\varphi}-\ell)(r-\ell)}f^{(\ell)}\left(X,Y\right)\ast g^{(\overline{\varphi}-\ell)}\left(X,Y\right).
\end{equation}

Now considering $\overline{\varphi}+1$ and for simplicity we write $f(X,Y;\lambda),~g(X,Y;\lambda)$ as $f,g$ we have
    \begin{align}
    \Big[f\ast g\Big]^{(\overline{\varphi}+1)} 
        & = \left[ \sum_{\ell=0}^{\overline{\varphi}}\bbinom{\overline{\varphi}}{\ell}b^{(\overline{\varphi}-\ell)(r-\ell)}f^{(\ell)}\ast g^{(\overline{\varphi}-\ell)}\right]^{(1)}
        \\
        & = \sum_{\ell=0}^{\overline{\varphi}}\bbinom{\overline{\varphi}}{\ell}b^{(\overline{\varphi}-\ell)(r-\ell)}\left[f^{(\ell)}\ast g^{(\overline{\varphi}-\ell)}\right]^{(1)}
        \\
        & = \sum_{\ell=0}^{\overline{\varphi}}\bbinom{\overline{\varphi}}{\ell}b^{(\overline{\varphi}-\ell)(r-\ell)}\left( b^{(r-\ell)}f^{(\ell)}\ast g^{(\overline{\varphi}-\ell+1)}+f^{(\ell+1)}\ast g^{(\overline{\varphi}-\ell)}\right)
        \\
        & = \sum_{\ell=0}^{\overline{\varphi}}\bbinom{\overline{\varphi}}{\ell}b^{(\overline{\varphi}-\ell+1)(r-\ell)}f^{(\ell)}\ast g^{(\overline{\varphi}-\ell+1)}
        \\
        &
        \hspace{1cm}+ \sum_{\ell=1}^{\overline{\varphi}+1}\bbinom{\overline{\varphi}}{\ell-1}b^{(\overline{\varphi}-\ell+1)(r-\ell+1)}f^{(\ell)}\ast g^{(\overline{\varphi}-\ell+1)}
        \\
        & = \bbinom{\overline{\varphi}}{0}b^{(\overline{\varphi}+1)r}f\ast g^{(\overline{\varphi}+1)}+ \sum_{\ell=1}^{\overline{\varphi}}\bbinom{\overline{\varphi}}{\ell}b^{(\overline{\varphi}+1-\ell)(r-\ell)}f^{(\ell)}\ast g^{(\overline{\varphi}-\ell+1)}
        \\
        & \hspace{1cm} + \bbinom{\overline{\varphi}}{\overline{\varphi}}b^{(\overline{\varphi}+1-\overline{\varphi}-1)(r-\overline{\varphi}-1+1)}f^{(\overline{\varphi}+1)}\ast g 
        \\
        & \hspace{1cm} + \sum_{\ell=1}^{\overline{\varphi}}\bbinom{\overline{\varphi}}{\ell-1}b^{(\overline{\varphi}+1-\ell)(r-\ell+1)}f^{(\ell)}\ast g^{(\overline{\varphi}-\ell+1)}
        \\
        & = b^{(\overline{\varphi}+1)r}f\ast g^{(\overline{\varphi}+1)} + f^{(\overline{\varphi}+1)}\ast g
        \\
        & \hspace{1cm} + \sum_{\ell=1}^{\overline{\varphi}} \left( \bbinom{\overline{\varphi}}{\ell} + b^{(\overline{\varphi}-\ell+1)}\bbinom{\overline{\varphi}}{\ell-1}\right) b^{(\overline{\varphi}-\ell+1)(r-\ell)}f^{(\ell)}\ast g^{(\overline{\varphi}-\ell+1)}
        \\
        \overset{\eqref{equation:Stepdown1}}&{=} \sum_{\ell=1}^{\overline{\varphi}}\bbinom{\overline{\varphi}+1}{\ell}b^{(\overline{\varphi}+1-\ell)(r-\ell)}f^{(\ell)}\ast g^{(\overline{\varphi}+1-\ell)}+ \bbinom{\overline{\varphi}+1}{0}b^{(\overline{\varphi}+1)(r)}f\ast g^{(\overline{\varphi}+1)}
        \\
        & \hspace{1cm} +\bbinom{\overline{\varphi}+1}{\overline{\varphi}+1} b^{(\overline{\varphi}-1-\overline{\varphi}-1)}f^{(\overline{\varphi}+1)}\ast g
        \\
        & = \sum_{\ell=0}^{\overline{\varphi}+1}\bbinom{\overline{\varphi}+1}{\ell}b^{(\overline{\varphi}+1-\ell)(r-\ell)}f^{(\ell)}\ast g^{(\overline{\varphi}+1-\ell)}.
    \end{align}
\end{proof}
 
\subsection{The \texorpdfstring{$b^{-1}$}{b-1}-Derivative}
Essentially, since the $b$-derivative finds a derivative with respect to $X$ it is natural to identify a comparable $b^{-1}$-derivative which can be used to develop a derivative with respect to $Y$.
\begin{defn}
For $b\neq 1$, the \textbf{\textit{$b^{-1}$-derivative}} at $Y\neq 0$ for a real-valued function $g(Y)$ is defined as
    \begin{equation}
        g^{\{1\}}\left(Y\right)=\dfrac{g\left(b^{-1}Y\right)-g\left(Y\right)}{(b^{-1}-1)Y}. 
    \end{equation}
For $\varphi\geq0$ we denote the $\varphi^{th}$ $b^{-1}$-derivative (with respect to $Y$) of $g(X,Y;\lambda)$ as $g^{\{\varphi\}}(X,Y;\lambda)$. The $0^{th}$ $b^{-1}$-derivative of $g(X,Y;\lambda)$ is $g(X,Y;\lambda)$. For any $a\in\mathbb{R},~Y\neq0$ and real-valued function $f(Y)$,
    \begin{equation}
        \Big[ f(Y)+ag(Y)\Big]^{\{1\}} = f^{\{1\}}(Y)+ag^{\{1\}}(Y).
    \end{equation}
\end{defn}
Again for the Hamming metric, we take the formal definition of a derivative and take the limit of the function as $b\to 1$, i.e. let $b^{-1}=1+h$, $h\in\mathbb{R}$ Then the derivative becomes,
\begin{equation}
    g^{\{1\}}(Y) = \lim_{h\to 0}\frac{g((1+h)Y)-g(Y)}{hY}
\end{equation}
and so again converts into the derivative in the usual sense of polynomials \cite[Problems (5), p98]{TheoryofError} with respect to $Y$. 

Similar to the $b$-derivative, since we have the definition of a derivative now with respect to $Y$ we can demonstrate some important results for homogeneous polynomials in general and the fundamental polynomials in particular. 
\begin{lem}\label{lemma:bderivatives} ~ \\
\begin{enumerate}
    \item For $0\leq \varphi \leq \ell,$ $\varphi\in\mathbb{Z}^{+}$ and $\ell \geq 0$,
        \begin{equation}
            \left(Y^\ell\right)^{\{\varphi\}} = b^{\varphi(1-\ell)+\sigma_{\varphi}}\beta_{b}(\ell,\varphi)Y^{\ell-\varphi}.
        \end{equation}
    \item The $\varphi^{th}$ $b^{-1}$-derivative of $g(X,Y;\lambda)=\displaystyle\sum_{i=0}^s g_i(\lambda) Y^{i}X^{s-i}$ is given by               \begin{equation}                                    \label{equation:bdervispoly}
                g^{\{\varphi\}}\left(X,Y;\lambda\right)=\displaystyle\sum_{i=\varphi}^{s}g_i(\lambda) b^{\varphi(1-i)+\sigma_{\varphi}} \beta_{b}(i,\varphi)Y^{i-\varphi}X^{s-i}.
        \end{equation}
    \item Also,
        \begin{align}\label{equation:mu-1bderiv}
            \mu^{[k]\{\varphi\}}(X,Y;\lambda) 
                & = b^{-\sigma_{\varphi}} \beta_{b}(k,\varphi)\gamma_{b,c}(\lambda,\varphi)\mu^{[k-\varphi]}(X,Y;\lambda-\varphi)
                \\
            \nu^{[k]\{\varphi\}}(X,Y;\lambda) 
                & = (-1)^{\varphi} \beta_{b}(k,\varphi)\nu^{[k-\varphi]}(X,Y;\lambda).\label{equation:nu-1bderiv}
\end{align}
\end{enumerate}
\end{lem}

\begin{proof}~\\
\begin{enumerate}
    \item[(1)] For $\varphi=1$ we have
        \begin{align}
            \left(Y^{\ell}\right)^{\{1\}} = \dfrac{\left(b^{-1}Y\right)^{\ell}-Y^{\ell}}{(b^{-1}-1)Y}
                & = \left(\dfrac{b^{-\ell}-1}{b^{-1}-1}\right)Y^{\ell-1} 
            \\
                & = \frac{bb^{-\ell}\left(1-b^{\ell}\right)}{1-b}Y^{\ell-1}
                \\
                & = b^{1-\ell}\beta_{b}(\ell,1)Y^{\ell-1}.
        \end{align}
    So the initial case holds. Assume the case for $\varphi =\overline{\varphi}$ holds. Then we have
        \begin{align}
            \left(Y^{\ell}\right)^{\{\overline{\varphi}+1\}}
                & = \left(b^{\overline{\varphi}(1-\ell)+\sigma_{\overline{\varphi}}}\beta_{b}(\ell,\overline{\varphi})Y^{\ell-\overline{\varphi}}\right)^{\{1\}}
                \\
                & = b^{\overline{\varphi}(1-\ell)+\sigma_{\overline{\varphi}}}\beta_{b}(\ell,\overline{\varphi})\frac{b^{-(\ell-\overline{\varphi})}Y^{\ell-\overline{\varphi}}-Y^{\ell-\overline{\varphi}}}{\left(b^{-1}-1\right)Y}
                \\
            & = b^{\overline{\varphi}(1-\ell)+\sigma_{\overline{\varphi}}}\beta_{b}(\ell,\varphi)b^{1-(\ell-\overline{\varphi})}\beta_{b}(\ell-\overline{\varphi},1)Y^{\ell-\overline{\varphi}-1}
            \\
            \overset{\eqref{equation:betaproperties}}&{=} b^{(\overline{\varphi}+1)(1-\ell)+\sigma_{\overline{\varphi}+1}}\beta_{b}(\ell,\overline{\varphi}+1)Y^{\ell-(\overline{\varphi}+1)}.
        \end{align}
    Thus the statement holds by induction. 
    \item[(2)] Now consider     $g(X,Y;\lambda)=\displaystyle\sum_{i=0}^s g_i (\lambda) Y^{i}X^{s-i}$. For $\varphi=1$ we have
        \begin{align}
            g^{\{1\}}\left(X,Y;\lambda\right) & = \left(\sum_{i=0}^s g_i(\lambda) Y^{i}X^{s-i}\right)^{\{1\}} 
            \\
            & = \sum_{i=0}^s g_i(\lambda) \left(Y^{i}\right)^{\{1\}}X^{s-i}
            \\
            & = \sum_{i=0}^s g_i(\lambda) b^{-i+1}\beta_{b}(i,1)Y^{i-1}X^{s-i}.
        \end{align}
    As $\beta_{b}(i,1)=0$ when $i=0$ and $\sigma_1=0$ then we have
        \begin{equation}
            g^{\{1\}}\left(X,Y;\lambda\right) = \sum_{i=1}^s g_i(\lambda) b^{1-i+\sigma_{1}}\beta_{b}(i,1)Y^{i-1}X^{s-i}.
        \end{equation}
    So the initial case holds. Now assume the case holds for $\varphi=\overline{\varphi}$ i.e.
    \begin{equation}
    g^{\{\overline{\varphi}\}}\left(X,Y;\lambda\right)=\displaystyle\sum_{i=\overline{\varphi}}^s g_i(\lambda) b^{\overline{\varphi}(1-i)+\sigma_{\overline{\varphi}}}\beta_{b}(i,\overline{\varphi})Y^{i-\overline{\varphi}}X^{s-i}.
    \end{equation}
    Then we have
    \begin{align}
        g^{\{\overline{\varphi}+1\}}\left(X,Y;\lambda\right) 
            & = \left(\sum_{i=\overline{\varphi}}^s g_i(\lambda) b^{\overline{\varphi}(1-i)+\sigma_{\overline{\varphi}}}\beta_{b}(i,\overline{\varphi})Y^{i-\overline{\varphi}}\right)^{\{1\}}X^{s-i}
            \\
            & = \sum_{i=\overline{\varphi}}^s g_i(\lambda) b^{\overline{\varphi}(1-i)+\sigma_{\overline{\varphi}}}\beta_{b}(i,\overline{\varphi})b^{-(i-\overline{\varphi}-1)}\beta_{b}(i-\overline{\varphi},1)Y^{i-\overline{\varphi}-1}X^{s-i}
            \\
            \overset{\eqref{equation:betafunction}}&{=} \sum_{i=\overline{\varphi}}^s g_i(\lambda) b^{(\overline{\varphi}+1)(1-i)+\sigma_{\overline{\varphi}}}\prod_{j=0}^{\overline{\varphi}-1}\frac{\left(b^{i-j}-1\right)\left(b^{i-\overline{\varphi}}-1\right)}{(b-1)(b-1)}Y^{i-\overline{\varphi}-1}X^{s-i}
            \\
            & = \sum_{i=\overline{\varphi}}^s g_i(\lambda) b^{(\overline{\varphi}+1)(1-i)+\sigma_{\overline{\varphi}+1}}\beta_{b}(i,\overline{\varphi}+1)Y^{i-\overline{\varphi}-1}X^{s-i}
            \\
            & = \sum_{i=\overline{\varphi}+1}^s g_i(\lambda) b^{(\overline{\varphi}+1)(1-i)+\sigma_{\overline{\varphi}+1}}\beta_{b}(i,\overline{\varphi}+1)Y^{i-\overline{\varphi}-1}X^{s-i}
    \end{align}
    since when $i=\overline{\varphi}$, $\beta_{b}(\overline{\varphi},\overline{\varphi}+1)=0$. So by induction Equation \eqref{equation:bdervispoly} holds.
    
    \item[(3)] Now consider $\mu^{[k]}(X,Y;\lambda)=\displaystyle\sum_{u=0}^k \mu_u(\lambda,k)Y^{u}X^{k-u}$ where $\mu_u(\lambda,k) = \displaystyle\bbinom{k}{u}\gamma_{b,c}(\lambda,u)$ as in Theorem \ref{theorem:bmu}. Then we have
    \begin{align}
        \mu^{[k]\{1\}}(X,Y;\lambda) 
            & = \left(\sum_{u=0}^k \mu_u(\lambda,k)Y^{u}X^{k-u}\right)^{\{1\}}
            \\
            \overset{\eqref{equation:bdervispoly}}&{=} \sum_{u=1}^k \mu_u(\lambda,k)b^{1-u}\beta_{b}(u,1)Y^{u-1}X^{k-u}\\
            & = \sum_{r=0}^{k-1} \mu_{r+1}(\lambda,k)b^{1-(r+1)}\beta_{b}(r+1,1)Y^{r+1-1}X^{k-r-1}
            \\
            & = \sum_{r=0}^{k-1} \bbinom{k}{r+1}\gamma_{b,c}(\lambda,r+1)b^{-r}\beta_{b}(r+1,1)Y^{r}X^{k-1-r}
            \\
            \overset{\eqref{equation:gammastepdown}\eqref{equation:beta1stepdown}}&{=} \sum_{r=0}^{k-1} \bbinom{k-1}{r}\frac{b^{k}-1}{b^{r+1}-1}\left(cb^\lambda-1\right)b^{r}b^{-r}\gamma_{b,c}(\lambda-1,r)
            \\
            & \hspace{1cm} \times \beta_{b}(r+1,1)Y^{r}X^{k-1-r}
            \\
            & = b^{-\sigma_1}\beta_{b}(k,1)\gamma_{b,c}(\lambda,1)\mu^{[k-1]}(X,Y;\lambda-1).
    \end{align}

    Now assume that the statement holds for $\varphi=\overline{\varphi}$. Then we have
    \begin{align}
        \mu^{[k]\{\overline{\varphi}+1\}}(X,Y;\lambda)
            & = \bigg[ b^{-\sigma_{\overline{\varphi}}}\beta_{b}(k,\overline{\varphi})\gamma_{b,c}(\lambda,\overline{\varphi})\mu^{[k-\overline{\varphi}]}(X,Y;\lambda-\overline{\varphi})\bigg]^{\{1\}}
            \\
        %
            & = b^{-\sigma_{\overline{\varphi}}}\beta_{b}(k,\overline{\varphi})\gamma_{b,c}(\lambda,\overline{\varphi})\left(\sum_{r=0}^{k-\overline{\varphi}}\bbinom{k-\overline{\varphi}}{r}\gamma_{b,c}(\lambda-\overline{\varphi},r)Y^{r}X^{k-\overline{\varphi}-r}\right)^{\{1\}}
            \\
            & = b^{-\sigma_{\overline{\varphi}}}\beta_{b}(k,\overline{\varphi})\gamma_{b,c}(\lambda,\overline{\varphi})\sum_{r=1
            }^{k-\overline{\varphi}}\bbinom{k-\overline{\varphi}}{r}\gamma_{b,c}(\lambda-\overline{\varphi},r)\left(Y^{r}\right)^{\{1\}}X^{k-\overline{\varphi}-r}
            \\
            & = b^{-\sigma_{\overline{\varphi}}}\beta_{b}(k,\overline{\varphi})\gamma_{b,c}(\lambda,\overline{\varphi})\sum_{u=0}^{k-\overline{\varphi}-1}\bbinom{k-\overline{\varphi}}{u+1}\gamma_{b,c}(\lambda-\overline{\varphi},u+1)b^{1-(u+1)}
            \\
            & \hspace{1cm} \times \beta_{b}(u+1,1)Y^{u+1-1}X^{k-\overline{\varphi}-u-1}
            \\
            \overset{\eqref{equation:gammastepdown}\eqref{equation:beta1stepdown}}&{=} b^{-\sigma_{\overline{\varphi}}}\beta_{b}(k,\overline{\varphi})\gamma_{b,c}(\lambda,\overline{\varphi})\sum_{u=0}^{k-(\overline{\varphi}+1)}\bbinom{k-\overline{\varphi}-1}{u}
            \\
            & \hspace{1cm} \times \frac{\left(b^{k-\overline{\varphi}}-1\right)\left(b^{u+1}-1\right)}{\left(b^{u+1}-1\right)(b-1)}b^{u}b^{-u}
            \\
            & \hspace{1cm} \times \left(cb^{\lambda-\overline{\varphi}}-1\right)\gamma_{b,c}(\lambda-(\overline{\varphi}+1),u)Y^{u}X^{k-(\overline{\varphi}+1)-u}
            \\
            & = b^{-\sigma_{\overline{\varphi}}}b^{-\overline{\varphi}}\gamma_{b,c}(\lambda,\overline{\varphi}+1)\beta_{b}(k,\overline{\varphi}+1)\mu^{[k-(\overline{\varphi}+1)]}(X,Y;\lambda-(\overline{\varphi}+1))
            \\
            & = b^{-\sigma_{\overline{\varphi}+1}}\gamma_{b,c}(\lambda,\overline{\varphi}+1)\beta_{b}(k,\overline{\varphi}+1)\mu^{[k-(\overline{\varphi}+1)]}(X,Y;\lambda-(\overline{\varphi}+1)).
\end{align}
As required.
Now consider $\nu^{[k]}(X,Y;\lambda)=\displaystyle\sum_{u=0}^k (-1)^u b^{u(u-1)}\displaystyle\bbinom{k}{u}Y^{u}X^{k-u}$ as defined in Theorem \ref{lemma:bnulemma}. Then we have
    \begin{align}
        \nu^{[k]\{1\}}(X,Y;\lambda) 
            & = \left(\sum_{u=0}^k (-1)^u b^{\sigma_{u}}\bbinom{k}{u}Y^{u}X^{k-u}\right)^{\{1\}}
            \\
            & = \sum_{r=0}^{k-1} (-1)^{r+1} b^{\sigma_{r+1}}b^{1-(r+1)}\bbinom{k}{r+1}\beta_{b}(r+1,1)Y^{r+1-1}X^{k-r-1}
            \\
        \overset{\eqref{equation:betafunction}}&{=} -\sum_{r=0}^{k-1} (-1)^{r} b^{\sigma_{r}}b^{r}b^{-r}\bbinom{k-1}{r}\frac{\left(b^{k}-1\right)\left(b^{r+1}-1\right)}{\left(b^{r+1}-1\right)\left(b-1\right)}Y^{r}X^{k-r-1}
            \\
            & = (-1)^{1}\beta_{b}(k,1)\nu^{[k-1]}(X,Y;\lambda).
\end{align}
Now assume that the statement holds for     $\varphi=\overline{\varphi}$. Then we have
    \begin{align}
        \nu^{[k]}(X,Y;\lambda)^{\{\overline{\varphi}+1\}} 
            & = \left[(-1)^{\overline{\varphi}}\beta_{b}(k,\overline{\varphi})\nu^{[k-\overline{\varphi}]}(X,Y;\lambda)\right]^{\{1\}}
            \\
            & = (-1)^{\overline{\varphi}}\beta_{b}(k,\overline{\varphi})\sum_{u=1}^{k-\overline{\varphi}} (-1)^u b^{\sigma_u}\bbinom{k-\overline{\varphi}}{u}\left(Y^{u}\right)^{\{1\}}X^{k-\overline{\varphi}-u}
            \\
            & = (-1)^{\overline{\varphi}}\beta_{b}(k,\overline{\varphi})\sum_{r=0}^{k-\overline{\varphi}-1} (-1)^{r+1} b^{\sigma_{r+1}}b^{-(r+1)+1}\bbinom{k-\overline{\varphi}}{r+1}
            \\
            & \hspace{1cm} \times \beta_{b}(r+1,1)Y^{r+1-1}X^{k-\overline{\varphi}-r-1}
            \\
            \overset{\eqref{equation:beta1stepdown}}&{=} (-1)^{\overline{\varphi}+1}\beta_{b}(k,\overline{\varphi})\sum_{r=0}^{k-\overline{\varphi}-1} (-1)^{r}b^{\sigma_{r}}\bbinom{k-\overline{\varphi}-1}{r}
            \\
            & \hspace{1cm} 
            \times \frac{\left(b^{k-\overline{\varphi}}-1\right)\left(b^{r+1}-1\right)}{\left(b^{r+1}-1\right)\left(b-1\right)}Y^{r}X^{k-\overline{\varphi}-1-r}
            \\
            & = (-1)^{\overline{\varphi}+1}\beta_{b}(k,\overline{\varphi}+1)\nu^{[k-(\overline{\varphi}+1)]}(X,Y;\lambda).
\end{align}
as required.
\end{enumerate}
\end{proof}

Now we need a few smaller lemmas in order to prove the Leibniz rule for the $b^{-1}$-derivative.

\begin{lem}\label{lemma:bminussimplification} Let
\begin{align}
    u\left(X,Y;\lambda\right) 
        & = \sum_{i=0}^r u_i(\lambda) Y^{i}X^{r-i}
        \\
     v\left(X,Y;\lambda\right) 
        & = \sum_{i=0}^s v_i(\lambda) Y^{i}X^{s-i}.
\end{align}
\begin{enumerate}
\item If $u_0(\lambda)=0$ then
\begin{equation}
    \frac{1}{Y}\Big[u\left(X,Y;\lambda\right)\ast v\left(X,Y;\lambda\right)\Big] = b^{s}\frac{u\left(X,Y;\lambda\right)}{Y}\ast v\left(X,Y;\lambda-1\right).\label{equation:binverseu=0}
\end{equation}
\item If $v_0(\lambda)=0$ then
\begin{equation}
    \frac{1}{Y}\Big[u\left(X,Y;\lambda\right)\ast v\left(X,Y;\lambda\right)\Big] = u\left(X, bY;\lambda\right)\ast \frac{v\left(X,Y;\lambda\right)}{Y}.\label{equation:binversev=0}
\end{equation}
\end{enumerate}
\end{lem}

\begin{proof}~\\
\begin{enumerate}
 \item[(1)] Suppose $u_0(\lambda)=0$. Then
\begin{equation}
    \frac{u\left(X,Y;\lambda\right)}{Y} = \sum_{i=0}^{r}u_i(\lambda) Y^{i-1}X^{r-i} = \sum_{i=0}^{r-1} u_{i+1}(\lambda)Y^{i}X^{r-i-1}
\end{equation}
Hence
\begin{align}
    b^{s}\frac{u\left(X,Y;\lambda\right)}{Y}~\ast & ~v\left(X,Y;\lambda-1\right) \\
        & = b^{s}\sum^{r+s-1}_{u=0}\left(\sum_{\ell=0}^u b^{\ell s}u_{\ell+1}(\lambda)v_{u-\ell}(\lambda-\ell-1)\right) Y^{u}X^{r+s-1-u}
        \\
        & = b^{s}\sum_{u=0}^{r+s-1}\left(\sum_{i=1}^{u+1}b^{(i-1)s}u_i(\lambda)v_{u-i+1}(\lambda-i)\right)Y^{u}X^{r+s-1-u}
        \\
        & = b^{s}\sum_{j=1}^{r+s}\left(\sum_{i=1}^{j}b^{(i-1)s}u_i(\lambda)v_{j-i}(\lambda-i)\right)Y^{j-1}X^{r+s-j}
        \\
        & = \frac{1}{Y}\sum_{j=0}^{r+s}\left(\sum_{i=0}^{j}b^{is}u_i(\lambda)v_{j-i}(\lambda-i)\right)Y^{j}X^{r+s-j}
        \\
        & = \frac{1}{Y}\left(u\left(X,Y;\lambda\right)\ast v\left(X,Y;\lambda\right)\right)
\end{align}
since when $j=0$, $\displaystyle\sum_{i=0}^{j}b^{is}u_i(\lambda)v_{j-i}(\lambda-i)=0$ as $u_0(\lambda)=0$.
 \item[(2)] Now if $v_0(\lambda)=0$, then
\begin{align}
    \frac{v\left(X,Y;\lambda\right)}{Y} 
        & = \sum_{j=1}^s v_j(\lambda)Y^{j-1}X^{s-j}
        \\
        & = \sum_{i=0}^{s-1} v_{i+1}(\lambda)Y^{i}X^{s-i-1}.
\end{align}
So, 
\begin{align}
    u\left(X,bY;\lambda\right) \ast \frac{v\left(X,Y;\lambda\right)}{Y} 
        & = \sum_{u=0}^{r+s-1}\left(\sum_{j=0}^{u} b^{j(s-1)}b^{j}u_j(\lambda)v_{u-j+1}(\lambda-j)\right)Y^{u}X^{r+s-1-u}
        \\
        & = \sum_{\ell=1}^{r+s}\left(\sum_{j=0}^{\ell-1} b^{js}u_j(\lambda)v_{\ell-j}(\lambda-j)\right)Y^{\ell-1}X^{r+s-\ell}
        \\
        & = \frac{1}{Y}\sum_{\ell=1}^{r+s}\left(\sum_{j=0}^{\ell} b^{js}u_j(\lambda)v_{\ell-j}(\lambda-j)\right)Y^{\ell}X^{r+s-\ell}
        \\
        & = \frac{1}{Y}\sum_{\ell=0}^{r+s}\left(\sum_{j=0}^{\ell} b^{js}u_j(\lambda)v_{\ell-j}(\lambda-j)\right)Y^{\ell}X^{r+s-\ell}
        \\
        & = \frac{1}{Y}\left(u\left(X,Y;\lambda\right)\ast v\left(X,Y;\lambda\right)\right)
\end{align}
since when $j=\ell$, $\displaystyle\sum_{i=0}^{j}b^{is}u_i(\lambda)v_{j-i}(\lambda-i)=0$ as $v_0(\lambda)=0$.
\end{enumerate}
\end{proof}

\begin{thm}[Leibniz rule for the $b^{-1}$-derivative]\label{bLiebnizbminusderiv}
For two homogeneous polynomials in $Y$, $f(X,Y;\lambda)$ and $g(X,Y;\lambda)$ with degrees $r$ and $s$ respectively, the $\varphi^{th}$ (for $\varphi\geq0$) $b^{-1}$-derivative of their $b$-product is given by

\begin{equation}\label{equation:bleibnizbminusderiv}
    \Big[ f\left(X,Y;\lambda\right)\ast g\left(X,Y;\lambda\right)\Big]^{\{\varphi\}} = \sum_{\ell=0}^\varphi \bbinom{\varphi}{\ell}b^{\ell(s-\varphi+\ell)}f^{\{\ell\}}\left(X,Y;\lambda\right)\ast g^{\{\varphi-\ell\}}\left(X,Y;\lambda-\ell\right).
\end{equation}
\end{thm}

\begin{proof}
For simplification we shall write $f(X,Y;\lambda),~g(X,Y;\lambda)$ as $f(Y;\lambda),~g(Y;\lambda)$. Now by differentiation we have
\begin{align}
    \Big[ f\left(Y;\lambda\right)\ast g\left(Y;\lambda\right)\Big]^{\{1\}} 
        & = \frac{f\left(b^{-1}Y;\lambda\right)\ast g\left(b^{-1}Y;\lambda\right)-f\left(Y;\lambda\right)\ast g\left(Y;\lambda\right)}{(b^{-1}-1)Y} 
        \\
        & = \frac{1}{(b^{-1}-1)Y} \bigg\{ f\left(b^{-1}Y;\lambda\right)\ast g\left(b^{-1}Y;\lambda\right)-f\left(b^{-1}Y;\lambda\right)\ast g\left(Y;\lambda\right) 
        \\
        & \hspace{1cm} + f\left(b^{-1}Y;\lambda\right)\ast g\left(Y;\lambda\right) - f\left(Y;\lambda\right)\ast g\left(Y;\lambda\right)\bigg\} 
        \\
        & = \frac{1}{(b^{-1}-1)Y}\bigg\{ f\left(b^{-1}Y;\lambda\right)\ast\left(g\left(b^{-1}Y;\lambda\right)-g\left(Y;\lambda\right)\right)\bigg\} 
        \\
        & \hspace{1cm} +\frac{1}{(b^{-1}-1)Y}\bigg\{\left(f\left(b^{-1}Y;\lambda\right)-f\left(Y;\lambda\right)\right)\ast g\left(Y;\lambda\right)\bigg\} 
        \\
        \overset{\eqref{equation:binversev=0}}&{=} f\left(Y;\lambda\right)\ast\frac{\left(g\left(b^{-1}Y;\lambda\right)-g\left(Y;\lambda\right)\right)}{\left(b^{-1}-1\right)Y} 
        \\
        &\hspace{1cm}\overset{\eqref{equation:binverseu=0}}{+}b^{s}\frac{\left(f\left(b^{-1}Y;\lambda\right)-f\left(Y;\lambda\right)\right)}{\left(b^{-1}-1\right)Y}\ast g\left(Y;\lambda-1\right) 
        \\
        & = f\left(Y;\lambda\right)\ast g^{\{1\}}\left(Y;\lambda\right) + b^{s} f^{\{1\}}\left(Y;\lambda\right)\ast g\left(Y;\lambda-1\right).\label{equation:lastlineleibnizb-1}
\end{align}
since $g(Y;\lambda)$ has the same degree as $g(b^{-1}Y;\lambda)$ and similarly, $f(Y;\lambda)$ has the same degree as $f(b^{-1}Y;\lambda)$. So the initial case holds. Assume the statement holds true for $\varphi=\overline{\varphi}$, i.e.

\begin{equation}
    \Big[ f\left(X,Y;\lambda\right)\ast g\left(X,Y;\lambda\right)\Big]^{\{\overline{\varphi}\}} = \sum_{\ell=0}^{\overline{\varphi}} \bbinom{\overline{\varphi}}{\ell}b^{\ell(s-\overline{\varphi}+\ell)}f^{\{\ell\}}\left(X,Y;\lambda\right)\ast g^{\{\overline{\varphi}-\ell\}}\left(X,Y;\lambda-\ell\right).
\end{equation}
Now considering $\overline{\varphi}+1$  and for simplicity we write $f(X,Y;\lambda),~g(X,Y;\lambda)$ as $f(\lambda),g(\lambda)$ we have

\begin{align}
    \Big[ f\left(\lambda\right)\ast g\left(\lambda\right)\Big]^{\{\overline{\varphi}+1\}} 
        & = \left[\sum_{\ell=0}^{\overline{\varphi}} \bbinom{\overline{\varphi}}{\ell}b^{\ell(s-\overline{\varphi}+\ell)}f^{\{\ell\}}\left(\lambda\right)\ast g^{\{\overline{\varphi}-\ell\}}\left(\lambda-\ell\right)\right]^{\{1\}}
        \\
    %
        \overset{\eqref{equation:lastlineleibnizb-1}}&{=} \sum_{\ell=0}^{\overline{\varphi}} \bbinom{\overline{\varphi}}{\ell}b^{\ell(s-\overline{\varphi}+\ell)}f^{\{\ell\}}\left(\lambda\right)\ast g^{\{\overline{\varphi}-\ell+1\}}\left(\lambda-\ell\right)
        \\
        & \hspace{1cm}+ \sum_{\ell=0}^{\overline{\varphi}} \bbinom{\overline{\varphi}}{\ell}b^{\ell(s-\overline{\varphi}+\ell)}b^{s-\overline{\varphi}+\ell}f^{\{\ell+1\}}\left(\lambda\right)\ast g^{\{\overline{\varphi}-\ell\}}\left(\lambda-\ell-1\right)
        \\
        & = \sum_{\ell=0}^{\overline{\varphi}} \bbinom{\overline{\varphi}}{\ell}b^{\ell(s-\overline{\varphi}+\ell)}f^{\{\ell\}}\left(\lambda\right)\ast g^{\{\overline{\varphi}-\ell+1\}}\left(\lambda-\ell\right)
        \\
        &  \hspace{1cm}+ \sum_{\ell=1}^{\overline{\varphi}+1} \bbinom{\overline{\varphi}}{\ell-1}b^{(\ell-1)(s-\overline{\varphi}+\ell-1)}b^{s-\overline{\varphi}+(\ell-1)}f^{\{\ell\}}\left(\lambda\right)\ast g^{\{\overline{\varphi}-\ell+1\}}\left(\lambda-\ell\right)
        \\
        & = f\left(\lambda\right)\ast g^{\{\overline{\varphi}+1\}}\left(\lambda\right)+\sum_{\ell=1}^{\overline{\varphi}} \bbinom{\overline{\varphi}}{\ell}b^{\ell(s-\overline{\varphi}+\ell)}f^{\{\ell\}}\left(\lambda\right)\ast g^{\{\overline{\varphi}-\ell+1\}}\left(\lambda-\ell\right)
        \\
        &  \hspace{1cm}+ \sum_{\ell=1}^{\overline{\varphi}} \bbinom{\overline{\varphi}}{\ell-1}b^{(\ell-1)(s-\overline{\varphi}+\ell-1)}b^{(s-\overline{\varphi}+(\ell-1))}f^{\{\ell\}}\left(\lambda\right)\ast g^{\{\overline{\varphi}-\ell+1\}}\left(\lambda-\ell\right)
        \\
        &  \hspace{1cm}+ \bbinom{\overline{\varphi}}{\overline{\varphi}}b^{(\overline{\varphi}+1)(s+1)}b^{-\overline{\varphi}-1} f^{\{\overline{\varphi}+1\}}\left(\lambda\right)\ast g\left(\lambda-(\overline{\varphi}+1)\right)
        \\
        & = f\left(\lambda\right)\ast g^{\{\overline{\varphi}+1\}}\left(\lambda\right)+ \sum_{\ell=1}^{\overline{\varphi}}\left(\bbinom{\overline{\varphi}}{\ell}+b^{-\ell}\bbinom{\overline{\varphi}}{\ell-1}\right)
        \\
        & \hspace{1cm} \times b^{\ell(s-\overline{\varphi}+\ell)}f^{\{\ell\}}\left(\lambda\right) \ast g^{\{\overline{\varphi}+1-\ell\}}\left(\lambda-\ell\right)
        \\
        &  \hspace{1cm} + b^{s(\overline{\varphi}+1)}f^{\{\overline{\varphi}+1\}}\left(\lambda\right) \ast g\left(\lambda-(\overline{\varphi}+1)\right)
        \\
        \overset{\eqref{equation:gaussiancoeffsx-1k-1}}&{=} f\left(\lambda\right)\ast g^{\{\overline{\varphi}+1\}}\left(\lambda\right) + \sum_{\ell=1}^{\overline{\varphi}}b^{-\ell}\bbinom{\overline{\varphi}+1}{\ell}b^{\ell(s-\overline{\varphi}+\ell)}f^{\{\ell\}}\left(\lambda\right)\ast g^{\{\overline{\varphi}+1-\ell\}}\left(\lambda-\ell\right)
        \\
        &  \hspace{1cm} +\bbinom{\overline{\varphi}+1}{\overline{\varphi}+1}b^{(\overline{\varphi}+1)(s-\overline{\varphi}-1+(\overline{\varphi}+1))}f^{\{\overline{\varphi}+1\}}\left(\lambda\right)\ast g^{\{\overline{\varphi}+1-(\overline{\varphi}+1)\}}\left(\lambda-(\overline{\varphi}+1)\right)
        \\
        & = \sum_{\ell=0}^{\overline{\varphi}+1}\bbinom{\overline{\varphi}+1}{\ell}b^{\ell(s-(\overline{\varphi}+1)+\ell)}f^{\{\ell\}}\left(\lambda\right) \ast g^{\{\overline{\varphi}+1-\ell\}}\left(\lambda-\ell\right)
\end{align}
as required.
\end{proof}
\pagebreak
\subsection{Evaluating the \texorpdfstring{$b$}{b}-Derivative and the \texorpdfstring{$b^{-1}$}{b-1}-Derivative}

At this point we need to introduce two lemmas which yield useful results when developing moments of the weight distribution.

\begin{lem}\label{lemma:bd'sandbetab's}
For $j,\ell\in\mathbb{Z}^+,~ 0\leq \ell \leq j$ and $X=Y=1$,
\begin{equation}\label{equation:bd'snadbetab's}
    \nu^{[j](\ell)}(1,1;\lambda)=\beta_{b}(j,j)\delta_{j\ell}.
\end{equation}
\end{lem}

\begin{proof}
Consider 
\begin{equation}
    \nu^{[j](\ell)}(X,Y;\lambda) 
        \overset{\eqref{equation:bnuderiv}}{=}  \beta_{b}(j,\ell)\nu^{[j-\ell]}(X,Y;\lambda) = \beta_{b}(j,\ell)\sum_{u=0}^{j-\ell}(-1)^ub^{\sigma_u}\bbinom{j-\ell}{u}Y^{u}X^{(j-\ell)-u}.
\end{equation}
So \begin{align}
    \nu^{[j](\ell)}(1,1;\lambda) 
        & = \beta_{b}(j,\ell)\sum_{u=0}^{j-\ell}(-1)^ub^{\sigma_u}\bbinom{j-\ell}{u}
        \\
        \overset{\eqref{equation:betabstartdifferent}}&{=} \beta_{b}(\ell,\ell)\bbinom{j}{\ell}\sum_{u=0}^{j-\ell}(-1)^ub^{\sigma_u}\bbinom{j-\ell}{u}
        \\
        \overset{\eqref{equation:gaussianxx-k}\eqref{equation:gaussianswapplaces}}&{=} \beta_{b}(\ell,\ell)\sum_{k=\ell}^{j}(-1)^{k-
        \ell}b^{\sigma_{k-\ell}}\bbinom{j}{k}\bbinom{k}{\ell}
        \\
        \overset{\eqref{equation:deltaijbs}}&{=}\beta_b(\ell,\ell)\delta_{\ell j}=\beta_{b}(j,j)\delta_{j\ell}.
\end{align}
\end{proof}

\begin{lem}\label{lemma:brhoandmu}
For any homogeneous polynomial, $\rho\left(X,Y;\lambda\right)$ and for any $s\geq 0$, \begin{equation}\label{equation:brhoandmu}
\left(\rho \ast \mu^{[s]}\right)\left(1,1;\lambda\right) =\left(cb^{\lambda}\right)^s\rho(1,1;\lambda).
\end{equation}
\end{lem}

\begin{proof}
Let $\rho\left(X,Y;\lambda\right)=\displaystyle\sum_{i=0}^r \rho_i(\lambda)Y^{i}X^{r-i}$, then from Theorem \ref{theorem:bmu}
\begin{equation}
    \mu^{[s]}(X,Y;\lambda) = \sum_{t=0}^s \mu^{[s]}_t(\lambda)Y^{t}X^{s-t} = \sum_{t=0}^{s}\bbinom{s}{t}\gamma_{b,c}(\lambda,t)Y^{t}X^{s-t} 
\end{equation}
and
\begin{equation}
    \left(\rho\ast \mu^{[s]}\right)(X,Y;\lambda) = \sum_{u=0}^{r+s}c_u(\lambda)Y^{u}X^{(r+s-u)}
\end{equation}
where
\begin{equation}
    c_u(\lambda)=\sum_{i=0}^u b^{is}\rho_i(\lambda)\mu^{[s]}_{u-i}(\lambda-i).
\end{equation}
Then, 
\begin{align}
    \left(\rho\ast \mu^{[s]}\right)(1,1;\lambda) 
        & = \sum_{u=0}^{r+s} c_u(\lambda)
        \\
        & 
        = \sum_{u=0}^{r+s}\sum_{i=0}^{u} b^{is}\rho_i(\lambda)\mu^{[s]}_{u-i}(\lambda-i)
        \\
        & = \sum_{j=0}^{r+s}b^{js}\rho_j(\lambda)\left(\sum_{k=0}^{r+s-j}\mu_k^{[s]}(\lambda-j)\right)
        \\
        & = \sum_{j=0}^{r}b^{js}\rho_j(\lambda)\left(\sum_{k=0}^s \mu_k^{[s]}(\lambda-j)\right)
        \\
        & = \sum_{j=0}^r b^{js} \rho_j(\lambda)\left(\sum_{k=0}^s {\bbinom{s}{k}}\gamma_{b,c}(\lambda-j,k)\right)
        \\
        \overset{\eqref{equation:producttosumgauss}}&{=} \sum_{j=0}^r b^{js} \rho_j(\lambda)\left(cb^{\lambda-j}\right)^s
        \\
        & = \left(cb^{\lambda }\right)^s\rho(1,1;\lambda)
\end{align}
since $\rho_j(\lambda)=0$ when $j > r$ and $\mu_k^{[s]}(\lambda-j)=0$ when $k>s$.
\end{proof}

\section{The \texorpdfstring{$b$}{b}-Moments of the Weight Distribution}\label{section:bmoments}

This final section develops a theory of $b$-moments analogous to \cite[Section 6]{friedlander2023macwilliams} and as before produces comparable formulas to the binomial moments in the Hamming case. Again the moments derived from the $b$-derivative and the $b^{-1}$-derivative are not exactly the same, as the first is using the derivative with respect to $X$ and the other is using the derivative with respect to $Y$. 
\subsection{Moments derived from the \texorpdfstring{$b$}{b}-Derivative}
In the first case we consider the moments of the weight distribution with respect to $X$.
\begin{prop}\label{prop:bmomentsbderiv}
For an $(\mathscr{X},R)$ $n$-class Krawtchouk association scheme, $0 \le \varphi \le n,$ and a linear code $\mathscr{C} \subseteq \mathscr{X}$, and its dual $\mathscr{C}^\perp\subseteq\mathscr{X}$ with weight distributions ${\boldsymbol{c}}=(c_0,\ldots,c_n)$ and ${\boldsymbol{ c'}}=(c'_0,\ldots,c'_n)$ respectively we have
\begin{equation}
    \sum_{i=0}^{n-\varphi}\bbinom{n-i}{\varphi}c_i = \frac{1}{|\mathscr{C}^\perp|}\left(cb^{n}\right)^{n-\varphi}\sum_{i=0}^{\varphi} \bbinom{n-i}{n-\varphi}c_i^{'}.
\end{equation}
\end{prop}

\begin{proof}
We apply Theorem \ref{ASSOCIATIONMacWilliams} to $\mathscr{C}^\perp$ to get
\begin{equation}
    W_{\mathscr{C}}^{S}(X,Y)=\frac{1}{\left| \mathscr{C}^\perp\right|}\overline{W}_{\mathscr{C}^\perp}^{S}\left( X +(cb^n-1)Y, X-Y\right)
\end{equation}
or equivalently
\begin{align}
    \sum_{i=0}^n c_i Y^{i}X^{n-i}
        & =\frac{1}{\left\vert \mathscr{C}^{\perp}\right|}\sum_{i=0}^n c_i'\left(X-Y\right)^{[i]}\ast \left[X +(cb^n-1)Y\right]^{[n-i]} 
        \\
        & = \frac{1}{\left\vert \mathscr{C}^{\perp}\right|} \sum_{i=0}^n c_i'\nu^{[i]}(X,Y;n)\ast \mu^{[n-i]}(X,Y;n). \label{equation:b}
\end{align}

For each side of Equation \eqref{equation:b}, we shall apply the $b$-derivative $\varphi$ times and then evaluate at $X=Y=1$.

For the left hand side, we obtain
\begin{equation}
    \left(\sum_{i=0}^n c_i Y^{i}X^{n-i}\right)^{(\varphi)}\overset{\eqref{equation:vthgeneralbderivative}}{=}
        \sum_{i=0}^{n-\varphi}c_i \beta_{b}(n-i,\varphi)Y^{i}X^{n-i-\varphi}.
\end{equation}
Setting $X=Y=1$ we then have
\begin{align}
    \sum_{i=0}^{n-\varphi} c_i\beta_{b}(n-i,\varphi)
        \overset{\eqref{equation:betabstartdifferent}}&{=} \sum_{i=0}^{n-\varphi}c_i\bbinom{n-i}{\varphi}\beta_{b}(\varphi,\varphi)
        \\ 
        & = \beta_{b}(\varphi,\varphi)\sum_{i=0}^{n-\varphi}c_i \bbinom{n-i}{\varphi}.
\end{align}

We now move on to the right hand side. For simplicity we write $\mu(X,Y;n)$ as $\mu$ and similarly $\nu(X,Y;n)$ as $\nu$. We have
\begin{align}
    \left(\frac{1}{\left\vert \mathscr{C}^\perp\right|}\sum_{i=0}^n c_i' \nu^{[i]}\ast \mu^{[n-i]}\right)^{(\varphi)} 
        \overset{\eqref{equation:bleibniz}}&{=} \frac{1}{\left\vert \mathscr{C}^\perp\right|}\sum_{i=0}^n c_i'\left(\sum_{\ell=0}^{\varphi}\bbinom{\varphi}{\ell}b^{(\varphi-\ell)(i-\ell)}\nu^{[i](\ell)}\ast \mu^{[n-i](\varphi-\ell)}\right)
        \\
        & = \frac{1}{\left\vert \mathscr{C}^\perp\right|}\sum_{i=0}^n c_i'\psi_i(X,Y;n)
\end{align}
where
\begin{equation}
    \psi_i(X,Y;n) 
         = \sum_{\ell=0}^{\varphi}\bbinom{\varphi}{\ell}b^{(\varphi-\ell)(i-\ell)}\nu^{[i](\ell)}(X,Y;n) \ast \mu^{[n-i](\varphi-\ell)}(X,Y;n).
\end{equation}
Then with $X=Y=1$,
\begin{align}
    \psi_i(1,1;n) 
        \overset{\eqref{equation:bmuderiv}}&{=} \sum_{\ell=0}^{\varphi} \bbinom{\varphi}{\ell}b^{(\varphi-\ell)(i-\ell)}\beta_{b}(n-i,\varphi-\ell)\left(\nu^{[i](\ell)}\ast \mu^{[n-i-\varphi+\ell]}\right)(1,1;n)
        \\
        \overset{\eqref{equation:brhoandmu}}&{=} \sum_{\ell=0}^{\varphi} \bbinom{\varphi}{\ell}b^{(\varphi-\ell)(i-\ell)}\beta_{b}(n-i,\varphi-\ell)\left(cb^{n}\right)^{n-i-(\varphi-\ell)}\nu^{[i](\ell)}(1,1;n)
        \\
        \overset{\eqref{equation:bd'snadbetab's}}&{=} \sum_{\ell=0}^{\varphi}b^{(\varphi-\ell)(i-\ell)}\bbinom{\varphi}{\ell}\beta_{b}(n-i,\varphi-\ell)\left(cb^{n}\right)^{n-i-(\varphi-\ell)}\beta_{b}(i,i)\delta_{i\ell}
        \\
    %
        \overset{\eqref{equation:betabstartdifferent}}&{=} \bbinom{\varphi}{i}\bbinom{n-i}{\varphi-1}\beta_{b}(\varphi-i,\varphi-i)\left(cb^{n}\right)^{n-\varphi}\beta_{b}(i,i)
        \\
        \overset{\eqref{equation:betabstartsame}}&{=} \bbinom{n-i}{\varphi-i}\left(cb^{n}\right)^{n-\varphi}\beta_{b}(\varphi,\varphi).
\end{align}

So

\begin{align}
    \frac{1}{|\mathscr{C}^\perp|}\sum_{i=0}^n c_i'\psi_i(1,1;n) 
        &  = \frac{1}{|\mathscr{C}^\perp|} \sum_{i=0}^{\varphi}c_i' \bbinom{n-i}{\varphi-i}\left(cb^{n}\right)^{n-\varphi}\beta_{b}(\varphi,\varphi)
        \\
        \overset{\eqref{equation:gaussianxx-k}}&{=} \beta_{b}(\varphi,\varphi)\frac{1}{|\mathscr{C}^\perp|}\left(cb^{n}\right)^{n-\varphi}\sum_{i=0}^{\varphi} c_i'\bbinom{n-i}{n-\varphi}.        
\end{align}
Combining the results for each side, and simplifying, we finally obtain
\begin{equation}
    \sum_{i=0}^{n-\varphi}c_i \bbinom{n-i}{\varphi} = \frac{1}{|\mathscr{C}^\perp|}\left(cb^{n}\right)^{n-\varphi}\sum_{i=0}^{\varphi} c_i'\bbinom{n-i}{n-\varphi}
\end{equation}
as required.
\end{proof}

\begin{note}
In particular, if $\varphi=0$ we have
\begin{equation}
    \sum_{i=0}^n c_i =\frac{\left(cb^{n}\right)^{n}}{|\mathscr{C}^\perp|}c_0' = \frac{\left(cb^{n}\right)^{n}}{|\mathscr{C}^\perp|}.
\end{equation}        
\end{note}

We can simplify Proposition \ref{prop:bmomentsbderiv} if $\varphi$ is less than the minimum distance of the dual code.

\begin{cor}\label{corrollary:bsimplificationpropbderiv}
Let $d_{S}'$ be the minimum distance of $\mathscr{C}^\perp$. If $0\leq \varphi < d_{S}'$ then
\begin{equation}
    \sum_{i=0}^{n-\varphi}\bbinom{n-i}{\varphi}c_i = \frac{1}{|\mathscr{C}^\perp|}\left(cb^{n}\right)^{n-\varphi} \bbinom{n}{\varphi}.
\end{equation}
\end{cor}

\begin{proof}
We have $c_0'=1$ and $c_1'=\ldots=c_\varphi'=0$.
\end{proof}


\subsection{Moments derived from the \texorpdfstring{$b^{-1}$}{b-1}-Derivative}

The next proposition once again relates the moments of the weight distribution of a linear code to those of its dual, this time using the $b^{-1}$-derivative of the MacWilliams Identity for a Krawtchouk association scheme. Here we must adapt the definition of $\delta(\lambda,\varphi,j)$ as in \cite[Lemma 6.3]{friedlander2023macwilliams} to make this definition applicable to all values of the parameter $c$ for a Krawtchouk association scheme.

\begin{lem}\label{lemma:bdeltas}

Let $\delta(\lambda,\varphi,j)=\displaystyle\sum_{i=0}^{j}(-1)^{i}\bbinom{j}{i}b^{\sigma_{i}}\gamma_{b,c}(\lambda-i,\varphi)$. Then for all $\lambda\in\mathbb{R},\varphi,j\in\mathbb{Z}$,

\begin{equation}\label{equation:bdeltas}
    \delta(\lambda,\varphi,j) = \prod_{i=0}^{j-1}\left(b^{\varphi}-b^i\right)\gamma_{b,c}(\lambda-j,\varphi-j)\left(cb^{\lambda-j}\right)^{j}.
\end{equation}

\end{lem}

\begin{proof}
Initial case: $j=0$.

\begin{align}
    \delta(\lambda,\varphi,0) & = \bbinom{0}{0}(-1)^{0}b^{\sigma_{0}}\gamma_{b,c}(\lambda,\varphi) = \gamma_{b,c}(\lambda,\varphi) =(\lambda,\varphi)\left(cb^{0(\lambda)}\right).
\end{align}
So the initial case holds. Now assume the case is true for $j=\overline{\jmath}$ and consider the $\overline{\jmath}+1$ case.

\begin{align}
    \delta(\lambda,\varphi,\overline{\jmath}+1) 
        & = \sum_{i=0}^{\overline{\jmath}+1}\bbinom{\overline{\jmath}+1}{i}(-1)^{i}b^{\sigma_i}\gamma_{b,c}(\lambda-i,\varphi)
        \\
        \overset{\eqref{equation:gaussiancoeffsx-1k-1}}&{=}\sum_{i=0}^{\overline{\jmath}+1}\left(b^{i}\bbinom{\overline{\jmath}}{i}+\bbinom{\overline{\jmath}}{i-1}\right)(-1)^{i}b^{\sigma_i}\gamma_{b,c}(\lambda-i,\varphi)
        \\
        & = \sum_{i=0}^{\overline{\jmath}}\bbinom{\overline{\jmath}}{i}(-1)^{i}b^{\sigma_i}b^{i}\gamma_{b,c}(\lambda-i,\varphi)+\sum_{i=0}^{\overline{\jmath}}\bbinom{\overline{\jmath}}{i}(-1)^{i+1}b^{\sigma_{i+1}}\gamma_{b,c}(\lambda-(i+1),\varphi)
        \\
        \overset{\eqref{equation:gammastepdown}}&{=}\sum_{i=0}^{\overline{\jmath}}\bbinom{\overline{\jmath}}{i}(-1)^{i}b^{i}b^{\sigma_i}\left(cb^{\lambda -i}-1\right)b^{\varphi-1}\gamma_{b,c}(\lambda-i-1,\varphi-1)
        \\
        & \hspace{1cm} \overset{\eqref{equation:gammastepdownsecond}}{-}\sum_{i=0}^{\overline{\jmath}}
        \bbinom{\overline{\jmath}}{i}(-1)^{i}b^{\sigma_{i+1}}\left(cb^{\lambda -i-1}-b^{\varphi-1}\right)\gamma_{b,c}(\lambda-i-1,\varphi-1)
        \\
        & = \sum_{i=0}^{\overline{\jmath}}\bbinom{\overline{\jmath}}{i}(-1)^{i}b^{\sigma_i}\gamma_{b,c}(\lambda-i-1,\varphi-1)(cb^{\lambda -1})\left(b^{\varphi}-1\right)
        \\
        & = cb^{\lambda -1}\left(b^{\varphi}-1\right)\delta(\lambda-1,\varphi-1,\overline{\jmath})
        \\
        & = cb^{\lambda -1}\left(b^{\varphi}-1\right)\prod_{i=0}^{\overline{\jmath}-1}\left(b^{\varphi-1}-b^i\right)c^{\overline{\jmath}}b^{\overline{\jmath}(\lambda-\overline{\jmath}-1)}\gamma_{b,c}(\lambda-\overline{\jmath}-1,\varphi-\overline{\jmath}-1)
        \\
        & = \left(b^{\varphi}-1\right)\prod_{i=0}^{\overline{\jmath}-1}\left(b^{\varphi-1}-b^i\right)(cb^{\lambda -1})c^{\overline{\jmath}}b^{\overline{\jmath}(\lambda-(\overline{\jmath}+1))}\gamma_{b,c}(\lambda-(\overline{\jmath}+1),\varphi-(\overline{\jmath}+1))
        \\
        & = \left(cb^{\lambda-(\overline{\jmath}+1)}\right)^{(\overline{\jmath}+1)}\prod_{i=0}^{\overline{\jmath}}\left(b^{\varphi}-b^i\right)\gamma_{b,c}(\lambda-(\overline{\jmath}+1),\varphi-(\overline{\jmath}+1))
\end{align}
since $\displaystyle\bbinom{\overline{\jmath}}{i-1}=0$ when $i=0$. Hence by induction the lemma is proved.
\end{proof}

\begin{lem}\label{lemma:bepsilons}
Let $\varepsilon(\Lambda,\varphi,i)=\displaystyle\sum_{\ell=0}^{i}\bbinom{i}{\ell}\bbinom{\Lambda-i}{\varphi-\ell}b^{\ell(\Lambda-\varphi)}(-1)^{\ell}b^{\sigma_{\ell}}\prod_{j=0}^{i-\ell-1}\left(b^{\varphi-\ell}-b^j\right)$. Then for all $\Lambda\in\mathbb{R},\varphi,i\in\mathbb{Z}$,
\begin{equation}\label{equation:bepsilons}
    \varepsilon(\Lambda,\varphi,i) = (-1)^{i}b^{\sigma_{i}}\bbinom{\Lambda-i}{\Lambda-\varphi}.
\end{equation}
\end{lem}

\begin{proof}
Initial case $i=0$,
\begin{equation}
    \varepsilon(\Lambda,\varphi,0) = \bbinom{0}{0}\bbinom{\Lambda}{\varphi}b^{0}(-1)^{0}b^{0} =\bbinom{\Lambda}{\varphi}
    (-1)^{0}b^{0}\bbinom{\Lambda}{\Lambda-\varphi} = \bbinom{\Lambda}{\varphi}.
\end{equation}
So the initial case holds. Now suppose the case is true when $i=\overline{\imath}$. Then

\begin{align}
    \varepsilon(\Lambda,\varphi,\overline{\imath}+1) 
        & = \sum_{\ell=0}^{\overline{\imath}+1}\bbinom{\overline{\imath}+1}{\ell}\bbinom{\Lambda-\overline{\imath}-1}{\varphi-\ell}b^{\ell(\Lambda-\varphi)}(-1)^{\ell}b^{\sigma_{\ell}}\prod_{j=0}^{\overline{\imath}-\ell}\left(b^{\varphi-\ell}-b^j\right)
        \\
        \overset{\eqref{equation:Stepdown1}}&{=} \sum_{\ell=0}^{\overline{\imath}+1}\bbinom{\overline{\imath}}{\ell}\bbinom{\Lambda-\overline{\imath}-1}{\varphi-\ell}b^{\ell(\Lambda-\varphi)}(-1)^{\ell}b^{\sigma_{\ell}}\prod_{j=0}^{\overline{\imath}-\ell}\left(b^{\varphi-\ell}-b^j\right)
        \\
        & \hspace{1cm} + \sum_{\ell=1}^{\overline{\imath}+1}b^{(\overline{\imath}+1-\ell)}\bbinom{\overline{\imath}}{\ell-1}\bbinom{\Lambda-\overline{\imath}-1}{\varphi-\ell}b^{\ell(\Lambda-\varphi)}(-1)^{\ell}b^{\sigma_{\ell}}\prod_{j=0}^{\overline{\imath}-\ell}\left(b^{\varphi-\ell}-b^j\right)
        \\
        & = A + B, \quad \text{say}.
\end{align}
Now
\begin{align}
    A 
        & = \left(b^{\varphi}-b^{\overline{\imath}}\right)\sum_{\ell=0}^{\overline{\imath}}\bbinom{\overline{\imath}}{\ell}\bbinom{\Lambda-\overline{\imath}-1}{\varphi-\ell}b^{\ell(\Lambda-1-\varphi)}(-1)^{\ell}b^{\sigma_{\ell}}\prod_{j=0}^{\overline{\imath}-\ell}\left(b^{\varphi-\ell}-b^j\right)
        \\
        & = \left(b^{\varphi}-b^{\overline{\imath}}\right)\varepsilon(\Lambda-1,\varphi,\overline{\imath})
        \\
        & = \left(b^{\varphi}-b^{\overline{\imath}}\right)(-1)^{\overline{\imath}}b^{\sigma_{\overline{\imath}}}\bbinom{\Lambda-\overline{\imath}-1}{\Lambda-1-\varphi}
\end{align}
and
\begin{align}
    B 
        & =\sum_{\ell=0}^{\overline{\imath}}b^{(\overline{\imath}-\ell)}\bbinom{\overline{\imath}}{\ell}\bbinom{\Lambda-\overline{\imath}-1}{\varphi-\ell-1}b^{(\ell+1)(\Lambda-\varphi)}(-1)^{\ell+1}b^{\sigma_{\ell+1}}\prod_{j=0}^{\overline{\imath}-\ell-1}\left(b^{\varphi-\ell-1}-b^j\right)
        \\
    %
        & = -b^{(\overline{\imath}+\Lambda-\varphi)}\varepsilon(\Lambda-1,\varphi-1,\overline{\imath})
        \\
        & = -b^{(\overline{\imath}+\Lambda                      -\varphi)}(-1)^{\overline{\imath}}b^{\sigma_{\overline{\imath}}}\bbinom{\Lambda-\overline{\imath}-1}{\Lambda-\varphi}.
\end{align}
So
\begin{align}
    \varepsilon(\Lambda,\varphi,\overline{\imath}+1) 
        & = A + B 
        \\
        & = (-1)^{\overline{\imath}}b^{\sigma_{\overline{\imath}}}\left\{\left(b^{\varphi}-b^{\overline{\imath}}\right)\bbinom{\Lambda-\overline{\imath}-1}{\Lambda-1-\varphi}-b^{(\overline{\imath}+\Lambda-\varphi)}\bbinom{\Lambda-\overline{\imath}-1}{\Lambda-\varphi}\right\}
        \\
        \overset{\eqref{equation:gaussianfracxk-1}}&{=} (-1)^{\overline{\imath}+1}b^{\sigma_{\overline{\imath}}}\left\{b^{\overline{\imath}+\Lambda-\varphi}\bbinom{\Lambda-\overline{\imath}-1}{\Lambda-\varphi}-\left(b^{\varphi}-b^{\overline{\imath}}\right)\frac{\left(b^{\Lambda-\varphi}-1\right)}{\left(b^{\varphi-\overline{\imath}}-1\right)}\bbinom{\Lambda-\overline{\imath}-1}{\Lambda-\varphi}\right\}
        \\
        & = (-1)^{\overline{\imath}+1}\bbinom{\Lambda-(\overline{\imath}+1)}{\Lambda-\varphi}b^{\sigma_{\overline{\imath}}}\left\{\frac{b^{\overline{\imath}+\Lambda-\varphi}\left(b^{\varphi-\overline{\imath}}-1\right)-\left(b^{\varphi}-b^{\overline{\imath}}\right)\left(b^{\Lambda-\varphi}-1\right)}{\left(b^{\varphi-\overline{\imath}}-1\right)}\right\}
        \\
    %
        & = (-1)^{\overline{\imath}+1}b^{\sigma_{\overline{\imath}+1}}\bbinom{\Lambda-(\overline{\imath}+1)}{\Lambda-\varphi}
\end{align}
as required.
\end{proof}

\begin{prop}\label{prop:bmomentsbminusderivative}
For an $(\mathscr{X},R)$ $n$-class Krawtchouk association scheme, $0 \le \varphi \le n$ and a linear code $\mathscr{C} \subseteq \mathscr{X}$ and its dual $\mathscr{C}^\perp\subseteq \mathscr{X}$ with weight distributions ${\boldsymbol{c}}=(c_0,\ldots,c_n)$ and ${\boldsymbol{ c'}}=(c'_0,\ldots,c'_n)$ respectively we have
\begin{equation}
    \sum_{i=\varphi}^n b^{\varphi(n-i)}\bbinom{i}{\varphi}c_i = \frac{1}{|\mathscr{C}^\perp|}\left(cb^{n}\right)^{n-\varphi}\sum_{i=0}^{\varphi}(-1)^{i}b^{\sigma_{i}}b^{i(\varphi-i)}\bbinom{n-i}{n-\varphi}\gamma_{b,c}(n-i,\varphi-i)c_{i}'.
\end{equation}
\end{prop}

\begin{proof}
As per Proposition \ref{prop:bmomentsbderiv}, we apply Theorem \ref{ASSOCIATIONMacWilliams} to $\mathscr{C}^{\perp}$ to obtain

\begin{equation}
    W_{\mathscr{C}}^{S}(X,Y) = \frac{1}{|\mathscr{C}^\perp|}\overline{W}_{\mathscr{C}^{\perp}}^{S}\left(X +(cb^n-1)Y,X-Y\right)
\end{equation}
or equivalently
\begin{align}
    \sum_{i=0}^n c_i Y^{i}X^{n-i} 
        & = \frac{1}{|\mathscr{C}^\perp|}\sum_{i=0}^{n}c_i'\left(X-Y\right)^{[i]}\ast \left(X +(cb^n-1)Y\right)^{[n-i]} 
        \\
        & = \frac{1}{|\mathscr{C}^\perp|}\sum_{i=0}^n c_i' \nu^{[i]}(X,Y;n)\ast \mu^{[n-i]}(X,Y;n). \label{equation:b-1}
\end{align}

For each side of Equation \eqref{equation:b-1}, we shall apply the $b^{-1}$-derivative $\varphi$ times and then evaluate at $X=Y=1$. i.e. 

\begin{equation}
    \left(\sum_{i=0}^n c_iY^{i}X^{n-i}\right)^{\{\varphi\}}  = \left(\frac{1}{|\mathscr{C}^\perp|}\sum_{i=0}^n c_i' \nu^{[i]}(X,Y;n)\ast \mu^{[n-i]}(X,Y;n)\right)^{\{\varphi\}}. \label{equation:bmomentswhatwewant}
\end{equation}

For the left hand side, we obtain
\begin{equation}\label{equation:bpropmomentsbminusLHS}
\begin{split}
    \left(\sum_{i=0}^n c_iY^{i}X^{n-i}\right)^{\{\varphi\}} 
        & = \sum_{i=\varphi}^n c_ib^{\varphi(1-i)+\sigma_{\varphi}}\beta_{b}(i,\varphi)Y^{i-\varphi}X^{n-i}
        \\
        \overset{\eqref{equation:betabstartdifferent}}&{=} \sum_{i=\varphi}^{n}c_i b^{\varphi(1-i)+\sigma_{\varphi}} \bbinom{i}{\varphi}\beta_{b}(\varphi,\varphi)Y^{i-\varphi}X^{n-i}.
    \end{split}
\end{equation}
Then using $X=Y=1$ gives
\begin{equation}
    \sum_{i=\varphi}^{n}c_i b^{\varphi(1-i)+\sigma_{\varphi}} \bbinom{i}{\varphi}\beta_{b}(\varphi,\varphi)Y^{i-\varphi}X^{n-i} = \sum_{i=\varphi}^n b^{\varphi(1-i)+\sigma_\varphi}\beta_{b}(\varphi,\varphi)\bbinom{i}{\varphi} c_i. \label{equation:bLHSmoment}
\end{equation}

We now move on to the right hand side. For simplicity we shall write $\mu(X,Y;n)$ as $\mu(n)$ and similarly $\nu(X,Y;n)$ as $\nu(n)$. We have,

 \begin{align}\label{equation:bpropmomentsbminusRHS}
     \Bigg(\frac{1}{|\mathscr{C}^\perp|}\sum_{i=0}^n & c_i' \nu^{[i]}(n) \ast \mu^{[n-i]}(n)\Bigg)^{\{\varphi\}} \\
        \overset{\eqref{equation:bleibnizbminusderiv}}&{=} \frac{1}{|\mathscr{C}^\perp|}\sum_{i=0}^n c_i' \left(\sum_{\ell=0}^\varphi\bbinom{\varphi}{\ell}b^{\ell(n-i-\varphi+\ell)}\nu^{[i]\{\ell\}}(n)\ast \mu^{[n-i]\{\varphi-\ell\}}(n-\ell)\right)
        \\
        & = \frac{1}{|\mathscr{C}^\perp|}\sum_{i=0}^n c_i' \psi_i(n)\label{equation:RHSpart1moments}
 \end{align}
say. Then,
\begin{align}
    \psi_i(n) 
        \overset{\eqref{equation:nu-1bderiv}\eqref{equation:mu-1bderiv}}&{=} \sum_{\ell=0}^{\varphi} \bbinom{\varphi}{\ell}b^{\ell(n-i-\varphi+\ell)}\left\{(-1)^\ell \beta_{b}(i,\ell)\nu^{[i-\ell]}(n)\right\}
        \\
        & \ast\left\{b^{-\sigma_{\varphi-\ell}}\beta_{b}(n-i,\varphi-\ell)\gamma_{b,c}(n-\ell,\varphi-\ell)\mu^{[n-i-\varphi+\ell]}(n-\varphi)\right\}.
\end{align}
Now let
\begin{equation}
    \Psi(X,Y;n-\varphi) = \nu^{[i-\ell]}(X,Y;n)\ast\gamma_{b,c}(n - \ell,\varphi-\ell)\mu^{[n-i-\varphi+\ell]}(X,Y;n-\varphi).
\end{equation}
Then we apply the $b$-product, reorder the summations  and set $X=Y=1$ to obtain
\begin{align}
    \Psi(1,1; & n-\varphi) \\
        & = \sum_{u=0}^{n-\varphi}\left[\sum_{p=0}^u b^{p(n-i-\varphi+\ell)}\nu_p^{[i-\ell]}(n)\gamma_{b,c}(n - \ell-p,\varphi-\ell)\mu^{[n-i-\varphi+\ell]}_{u-p}(n-\varphi-p)\right]
        \\
        & = \sum_{r=0}^{i-\ell}b^{r(n-i-\varphi+\ell)}\nu_r^{[i-\ell]}(n)\gamma_{b,c}(n - \ell-r,\varphi-\ell)\left[ \sum_{w=0}^{n-i-\varphi+\ell}\mu_w^{[n-i-\varphi+\ell]}(n-\varphi-r)\right]
        \\
        \overset{\eqref{equation:producttosumgauss}}&{=} \sum_{r=0}^{i-\ell}b^{r(n-i-\varphi+\ell)}\left(cb^{n-\varphi-r}\right)^{(n-i-\varphi+\ell)}\nu_r^{[i-\ell]}(n)\gamma_{b,c}(n - \ell-r,\varphi-\ell)
        \\
        & = \left(cb^{n-\varphi}\right)^{n-i-\varphi+\ell}\sum_{r=0}^{i-\ell}(-1)^r b^{\sigma_{r}}\bbinom{i-\ell}{r}\gamma_{b,c}(n - \ell-r,\varphi-\ell)
        \\
        & = \left(cb^{n-\varphi}\right)^{n-i-\varphi+\ell}\delta(n-\ell,\varphi-\ell, i-\ell)
        \\
       \overset{\eqref{equation:bdeltas}}&{=} \left(cb^{n-\varphi}\right)^{n-i-\varphi+\ell}\left(cb^{n-i}\right)^{i-\ell}\prod_{j=0}^{i-\ell-1}\left(b^{\varphi-\ell}-b^j\right)\gamma_{b,c}(n-i,\varphi-i)
        \\
        & = c^{n-\varphi}b^{(n-\varphi)(n-i-\varphi+\ell)}b^{(i-\ell)(n-i)}\prod_{j=0}^{i-\ell-1}\left(b^{\varphi-\ell}-b^j\right)\gamma_{b,c}(n-i,\varphi-i).
\end{align}
Noting that $b^{\ell(n-i-\varphi+\ell)}b^{-\sigma_{\varphi-\ell}}=b^{\ell(n-i)}b^{-\sigma_{\varphi}}b^{\sigma_\ell}$ we get
\begin{align}
    \psi_i(1,1;n) 
        & = \sum_{\ell=0}^{\varphi}(-1)^\ell\bbinom{\varphi}{\ell}b^{\ell(n-i-\varphi+\ell)} b^{-\sigma_{\varphi-\ell}}\beta_{b}(i,\ell)\beta_{b}(n-i,\varphi-\ell)\Psi(1,1;n-\varphi)
        \\
        \overset{\eqref{equation:betabstartsame}}&{=} \sum_{\ell=0}^{\varphi}(-1)^\ell\bbinom{\varphi}{\ell}b^{\ell(n-i-\varphi+\ell)} b^{-\sigma_{\varphi-\ell}}\bbinom{i}{\ell}\beta_{b}(\ell,\ell)
        \\
        & \hspace{1cm} \times \bbinom{n-i}{\varphi-\ell}\beta_{b}(\varphi-\ell,\varphi-\ell)\Psi(1,1;n-\varphi)
        \\
        \overset{\eqref{equation:betabstartsame}}&{=} 
        b^{-\sigma_{\varphi}}\beta_{b}(\varphi,\varphi)\sum_{\ell=0}^{i}(-1)^{\ell}b^{\ell(n-i)}b^{\sigma_{\ell}}\bbinom{i}{\ell}\bbinom{n-i}{\varphi-\ell}\Psi(1,1;n-\varphi).
\end{align}
Writing that
\begin{align}
    b^{-\sigma_{\varphi}}b^{\ell(n-i)}b^{(n-\varphi)(n-\varphi-i+\ell)}b^{(i-\ell)(n-i)} 
        & = b^{\sigma_{\varphi}}b^{\varphi(1-n)}b^{n(n-\varphi)}b^{\ell(n-\varphi)}b^{i(\varphi-i)}
        \\
        & = b^{\theta}b^{\ell(n-\varphi)}
\end{align}
we get
\begin{align}
    \psi_i(1,1;n) 
        & = c^{n-\varphi}b^{\theta}\beta_{b}(\varphi,\varphi)\gamma_{b,c}(n-i,\varphi-i)\sum_{\ell=0}^i(-1)^\ell b^{\ell(n-\varphi)} b^{\sigma_{\ell}}\bbinom{i}{\ell}\bbinom{n-i}{\varphi-\ell}\prod_{j=0}^{i-\ell-1}\left(b^{\varphi-\ell}-b^j\right) 
        \\
        \overset{\eqref{equation:bepsilons}}&{=} c^{n-\varphi}b^{\theta}b^{\sigma_{i}}\beta_{b}(\varphi,\varphi)\bbinom{n-i}{n-\varphi}\gamma_{b,c}(n-i,\varphi-i)\label{equation:momentsRHSpsi}.
\end{align}

Substituting the results from \eqref{equation:bLHSmoment}, \eqref{equation:RHSpart1moments} and \eqref{equation:momentsRHSpsi} we have
\begin{equation}
    \sum_{i=\varphi}^n b^{\varphi(1-i)+\sigma_\varphi}\beta_{b}(\varphi,\varphi)\bbinom{i}{\varphi} c_i = \frac{1}{|\mathscr{C}^\perp|}\sum_{i=0}^n c_i'(-1)^{i}c^{n-\varphi}b^{\theta}b^{\sigma_{i}}\beta_{b}(\varphi,\varphi)\bbinom{n-i}{n-\varphi}\gamma_{b,c}(n-i,\varphi-i).
\end{equation}
Thus cancelling and rearranging gives,
\begin{equation}
    \sum_{i=\varphi}^n b^{\varphi(n-i)} \bbinom{i}{\varphi}c_i  = \frac{\left(cb^{n}\right)^{n-\varphi}}{|\mathscr{C}^\perp|}\sum_{i=0}^\varphi  (-1)^{i} b^{\sigma_i}b^{i(\varphi-i)}\bbinom{n-i}{n-\varphi}\gamma_{b,c}(n-i,\varphi-i)c_i'
\end{equation}
as required.
\end{proof}

We can simplify Proposition \ref{prop:bmomentsbminusderivative} if $\varphi$ is less than the minimum distance of the dual code. Also we can introduce the \textbf{\textit{dual diameter}}, $\varrho_{S}'$, defined as the maximum distance between any two codewords of the dual code and simplify Proposition \ref{prop:bmomentsbminusderivative} further.

\begin{cor}
If $0\leq \varphi < d_{S}'$ then
\begin{equation}
    \sum_{i=\varphi}^n b^{\varphi(n-i)}\bbinom{i}{\varphi}c_i 
        = 
    \frac{1}{|\mathscr{C}^\perp|}\left(cb^{n}\right)^{n-\varphi}\bbinom{n}{\varphi}\gamma_{b,c}(n,\varphi).
\end{equation}
For $\varrho_{S}'< \varphi \leq n$ then
\begin{equation}
   \sum_{i=0}^{\varphi}(-1)^{i}b^{\sigma_{i}}b^{i(\varphi-i)}\bbinom{n-i}{n-\varphi}\gamma_{b,c}(n-i,\varphi-i)c_{i}
        =
    0.
\end{equation}
Explicitly for the Hamming association scheme, when $b=1$ and $c=q$ we have for $\varrho_{S}'< \varphi \leq n$,
\begin{equation}
   \sum_{i=0}^{\varphi}(-1)^{i}\binom{n-i}{n-\varphi}\left(q-1\right)^{\varphi-i}c_{i}
        =
    0.
\end{equation}
Moreover for $\varphi=n$, 
\begin{equation}
    \sum_{i=0}^{n}(-1)^{i}(q-1)^{n-i}c_i = 0.
\end{equation}
\end{cor}

\begin{proof}
First consider $0\leq\varphi< d_{S}'$, then $c'_0=1$, $c_1'=\ldots=c_\varphi'=0$. Also since $\displaystyle\bbinom{n}{n-\varphi}=\bbinom{n}{\varphi}$ the statement holds. Now if $\varrho_{S}'<\varphi\leq n$ then applying Proposition \ref{prop:bmomentsbminusderivative} to $\mathscr{C}^\perp$ gives
\begin{equation}
    \sum_{i=\varphi}^n b^{\varphi(n-i)}\bbinom{i}{\varphi}c'_i 
        = 
    \frac{1}{|\mathscr{C}|}\left(cb^{n}\right)^{n-\varphi}\sum_{i=0}^{\varphi}(-1)^{i}b^{\sigma_{i}}b^{i(\varphi-i)}\bbinom{n-i}{n-\varphi}\gamma_{b,c}(n-i,\varphi-i)c_{i}.
\end{equation}
So using $c_\varphi'=\ldots=c_n'=0$ we get
\begin{equation}
    0 = \sum_{i=0}^{\varphi}(-1)^{i}b^{\sigma_{i}}b^{i(\varphi-i)}\bbinom{n-i}{n-\varphi}\gamma_{b,c}(n-i,\varphi-i)c_{i}
\end{equation}
as required. For the Hamming association scheme, we use that $b=1$, $c=q$ and the $b$-nary Gaussian coefficients become the usual binomal coefficients and we have immediately
\begin{equation}
    0 = \sum_{i=0}^{\varphi}(-1)^{i}\binom{n-i}{n-\varphi}(q-1)^{\varphi-i}c_{i}.
\end{equation}
Moreover when $\varphi=n$,
\begin{equation}
    \sum_{i=0}^{n}(-1)^{i}\binom{n-i}{0}(q-1)^{\varphi-i}c_{i}
     = \sum_{i=0}^{n}(-1)^{i}c_{i} = 0.
\end{equation}
\end{proof}

\subsection{Maximum Distance Codes in the Association Scheme}
As an application for the MacWilliams Identity, we can derive an explicit form of the coefficients of the weight distribution for an $(\mathscr{X},R)$ $n$-class association scheme for maximal distance codes. This generalises the results for MDS codes \cite[Theorem 6, Chapter 11]{TheoryofError}, MRD codes \cite[Proposition 9]{gadouleau2008macwilliams}, MSRD codes in \cite[Proposition 6.8]{friedlander2023macwilliams} and similar MHRD codes.

Firstly a lemma that will be needed.
\begin{lem}\label{lemma:bsequences}
If $x_0,x_1,\ldots,x_\ell$ and $y_0,y_1,\ldots,y_\ell$ are two sequences of real numbers and if
\begin{equation}
    x_j = \sum_{i=0}^{j}\bbinom{\ell-i}{\ell-j}y_i
\end{equation}
for $0\leq j\leq \ell$, then

\begin{equation}
    y_i = \sum_{j=0}^i (-1)^{i-j}b^{\sigma_{i-j}}\bbinom{\ell-j}{\ell-i}x_j
\end{equation}
for $0\leq i\leq \ell$.
\end{lem}
\begin{proof}
For $0\leq i \leq \ell$, 
\begin{align}
    \sum_{j=0}^i (-1)^{i-j}b^{\sigma_{i-j}}\bbinom{\ell-j}{\ell-i}x_j 
        & = \sum_{j=0}^i(-1)^{i-j}b^{\sigma_{i-j}}\bbinom{\ell-j}{\ell-i}\left(\sum_{k=0}^j \bbinom{\ell-k}{\ell-j}y_k\right)
        \\
        & = \sum_{k=0}^i \sum_{j=k}^i (-1)^{i-j}b^{\sigma_{i-j}}
        \bbinom{\ell-j}{\ell-i}\bbinom{\ell-k}{\ell-j}y_k
        \\
        & = \sum_{k=0}^i y_k \left(\sum_{s=\ell-i}^{\ell-k}(-1)^{i-\ell+s}b^{\sigma_{i-\ell+s}}\bbinom{s}{\ell-i}\bbinom{\ell-k}{s}\right)
        \\
        \overset{\eqref{equation:deltaijbs}}&{=} \sum_{k=0}^i y_k \delta_{ik}
        \\
        & = y_i
\end{align}
as required.
\end{proof}

Before going any further we need some restrictions on the codes we consider to be able to use the following proposition. We are only considering $(\mathscr{X},R)$ $n$-class Krawtchouk association schemes. From there we are restricted to linear codes $\mathscr{C}\subseteq\mathscr{X}$ with minimum distance $d_{S}$ and their dual codes $\mathscr{C}^\perp\subseteq\mathscr{X}$ with minimum distance $d'_{S}$ such that $d_S+d'_S = n+2$. This restriction is necessary since the ``first pair of universal bounds" \cite[Section IV.F]{InfoTheoryDelsarte} is met in equality if and only if $d_S+d'_S=n+2$. We call codes that meet these bounds \textbf{\textit{maximal}} codes. More details on these ``universal bounds", which are the equivalent Singleton bounds, for any $P$-polynomial scheme can be found in \cite[Section IV.F]{InfoTheoryDelsarte}.

\begin{prop}
For an $(\mathscr{X},R)$ $n$-class Krawtchouk association scheme let $\mathscr{C}\subseteq\mathscr{X}$ be a maximal linear code with weight distribution $\boldsymbol{c}=(c_0,\ldots,c_n)$ and minimum distance $d_S$. Let the dual of $\mathscr{C}$ be the maximal linear code $\mathscr{C}^\perp$ with minimum distance $d'_S=n-d_S+2$. Then we have $c_0=1$ and for $0\leq \omega \leq n-d_{S}$,
\begin{equation}
    c_{d_{S}+\omega} = \sum_{i=0}^\omega (-1)^{\omega-i} b^{\sigma_{\omega-i}}\bbinom{d_S+\omega}{d_S+i}\bbinom{n}{d_S+\omega}\left( \frac{cb^{n\left(d_{S}+i\right)}}{|\mathscr{C}^\perp|}-1\right).
\end{equation}
\end{prop}

\begin{proof}
Now from Corollary \ref{corrollary:bsimplificationpropbderiv} we have
\begin{equation}
    \sum_{i=0}^{n-\varphi}\bbinom{n-i}{\varphi}c_i  = \frac{1}{|\mathscr{C}^\perp|}\left(cb^n\right)^{n-\varphi}\bbinom{n}{\varphi}
\end{equation}
for $0\leq \varphi < d_{S}'$. Now since we have a linear code $\mathscr{C}$ which is maximal, with minimum distance $d_{S}$ and we have  $\mathscr{C}^\perp$ which is also maximal with minimum distance $d_{S}'=n-d_{S}+2$, Corollary \ref{corrollary:bsimplificationpropbderiv} holds for $0\leq \varphi\leq n-d_{S}=d_{S}'-2$. We therefore have $c_0=1$ and $c_1=c_2=\ldots=c_{d_{S}-1}=0$ and setting $\varphi=n-d_{S}-j$ for $0\leq j\leq n-d_{S}$ we obtain

\begin{align}
    \bbinom{n}{n-d_S-j} + \sum_{i=d_{S}}^{d_{S}+j}\bbinom{n-i}{n-d_S-j}c_i & = \frac{1}{|\mathscr{C}^\perp|}\left(cb^n\right)^{d_{S}+j}\bbinom{n}{n-d_S-j}\\
    \sum_{\omega=0}^j \bbinom{n-d_S-\omega}{n-d_S-j}c_{\omega+d_{S}} & = \bbinom{n}{n-d_S-j}\left(\frac{\left(cb^n\right)^{d_{S}+j}}{|\mathscr{C}^\perp|}-1\right).
\end{align}
Applying Lemma \ref{lemma:bsequences} with $\ell = n-d_{S}$ and $b_\omega = c_{\omega+d_{S}}$ then setting 
\begin{equation}
    a_j = \bbinom{n}{n-d_S-j}\left(\frac{\left(cb^n\right)^{d_{S}+j}}{|\mathscr{C}^\perp|}-1\right)
\end{equation}
gives
\begin{equation}
     \sum_{\omega=0}^j \bbinom{n-d_S-\omega}{n-d_S-j}b_\omega = a_j
\end{equation}
and so
\begin{align}
    b_\omega = c_{\omega+d_{S}} & = \sum_{i=0}^\omega (-1)^{\omega-i}b^{\sigma_{\omega-i}}\bbinom{n-d_S-i}{n-d_S-\omega}a_i\\
    & = \sum_{i=0}^\omega (-1)^{\omega-i}b^{\sigma_{\omega-i}}\bbinom{n-d_S-i}{n-d_S-\omega}\bbinom{n}{n-d_S-i}\left(\frac{\left(cb^n\right)^{d_{S}+i}}{|\mathscr{
    C}^\perp|}-1\right).
\end{align}
But we have
\begin{align}
    \bbinom{n-d_S-i}{n-d_S-\omega}\bbinom{n}{n-d_S-i} \overset{\eqref{equation:gaussianxx-k}}&{=} \bbinom{n-(d_S+i)}{n-(d_S+\omega)}\bbinom{n}{d_S+i}\\
    \overset{\eqref{equation:gaussianswapplaces}}&{=} \bbinom{d_S+\omega}{d_S+i}\bbinom{n}{n-(d_S+\omega)}\\
     \overset{\eqref{equation:gaussianxx-k}}&{=} \bbinom{d_S+\omega}{d_S+i}\bbinom{n}{d_S+\omega}.
\end{align}
Therefore
\begin{equation}
    c_{\omega+d_{S}} = \sum_{i=0}^\omega (-1)^{\omega-i}q^{2\sigma_{\omega-i}}\bbinom{d_S+\omega}{d_S+i}\bbinom{n}{d_S+\omega}\left(\frac{q^{m(d_{S}+i)}}{|\mathscr{
    C}^\perp|}-1\right)
\end{equation}
as required.
\end{proof}

\section{Individual Association Schemes}\label{section:individual}

Since now we have a general theory for the Krawtchouk association schemes, it is useful to now compare these results to those obtained for individual known schemes, namely, the Hamming, rank, skew rank and Hermitian association schemes. 

\subsection{The Parameters}

In more detail, for the $n$-class Hamming scheme we take $b=1$ and $c=q$ which gives us the MacWilliams Identity as:
\begin{equation}
W_{\mathscr{C}^\perp}^{H}(X,Y)
            = \frac{1}{|\mathscr{C}|}\sum_{i=0}^n c_i(X-Y)^{i}\left(X+\left(q-1\right)Y\right)^{n-i}
\end{equation}
where $W_{\mathscr{C}^\perp}^{H}(X,Y)$ is the Hamming weight enumerator of $X$ and $Y$ and $\ast$ becomes usual polynomial multiplication.

For the rank association scheme over $\mathbb{F}_{q}^{m\times n}$ (or $\mathbb{F}_{q^m}^{n}$), studied by Gadouleau and Yan \cite{gadouleau2008macwilliams}, we know what the MacWilliams Identity should look like. So taking $b=q$ and $c=q^{m-n}$ we obtain, 
\begin{equation}
    W_{\mathscr{C}^\perp}^{R}(X,Y)
             = \frac{1}{|\mathscr{C}|}\sum_{i=0}^n c_i(X-Y)^{[i]}\ast\left(X+\left(q^m-1\right)Y\right)^{[n-i]}
\end{equation}
as required, with the correct ``product", $\ast$. 

For the skew rank association scheme over $\mathbb{F}_q^{t\times t}$, studied by Friedlander et al. \cite{friedlander2023macwilliams} we take $b=q^2$ and $c$ to be $q$ if $t$ is odd, and $q^{-1}$ if $t$ is even. So we have the MacWilliams Identity as,
\begin{equation}
    W_{\mathscr{C}^\perp}^{SR}(X,Y)
             = \frac{1}{|\mathscr{C}|}\sum_{i=0}^n c_i(X-Y)^{[i]}\ast\left(X+\left(q^m-1\right)Y\right)^{[n-i]}
\end{equation}
where $n=\lfloor\frac{t}{2}\rfloor$, and $m=\frac{t(t-1)}{2n}$ as required, with the correct ``product", $\ast$ again. 

Finally for the Hermitian association scheme over $\mathbb{F}_{q^2}^{t\times t}$, we take $b=-q$ and $c=-1$ and we obtain the previously unpublished MacWilliams Identity
\begin{equation}
    W_{\mathscr{C}^\perp}^{HR}(X,Y)
             = \frac{1}{|\mathscr{C}|}\sum_{i=0}^t c_i(X-Y)^{[i]}\ast\left(X+\left(-(-q)^t-1\right)Y\right)^{[t-i]}
\end{equation}
with yet again a slightly different product, $\ast$. 

These parameters are summarised in Table \ref{tab:general2} below. 
\begin{table}[H]
    \centering\renewcommand{\arraystretch}{1.2}
    \begin{tabular}{|c|c|c|c|c|}
        \hline
            Name & $b$ & $c$ & $\lambda$ & $cb^{\lambda}$
            \\
        \hline\rowcolor{mycolor!10}
            Hamming & $1$ & $q$ & $n$ & $q$ 
            \\
        \hline\rowcolor{mycolor!20}
            Bilinear & &  &  & 
            \\ \rowcolor{mycolor!20}
            Gabidulin & \multirow{-2}{0.8em}{$q$}  & \multirow{-2}{1.5em}{$q^{m-n}$} & \multirow{-2}{0.5em}{$n$} & \multirow{-2}{1.8em}{$q^{m}$}
            \\
        \hline \rowcolor{mycolor!30}
            Skew & $q^{2}$ & t odd - $q$, t even - $q^{-1}$ & $n$ & $cq^{2n}$
            \\
        \hline \rowcolor{mycolor!40}
            Hermitian & $-q$ & $-1$ & $t$ & $-(-q)^{t}$
            \\
        \hline
    \end{tabular}
    \caption{Proposed generalised parameters}
    \label{tab:general2}
\end{table}

\subsection{The Overview}

To summarise the results of the Krawtchouk association schemes, Tables \ref{tab:general1} and \ref{tab:general3} highlight the features we need in each one. At a glance there are evident similarities but significant differences. For each case we have listed the associated metric, the number of classes in the relevant association scheme, the underlying space, the homogeneous polynomial used in the MacWilliams Identity, the ``fundamental polynomial", the known eigenvalues of the scheme and its valencies. In all cases $q$ is a power of a prime. The final row is the accumulation of the results from analysing the ways in which the differences can be assimilated in a general theory.
\begin{table}[H]
    \centering\renewcommand{\arraystretch}{1.4}
    \begin{tabular}{|c|c|c|c|c|c|c|}
        \hline
            Name & Metric & Class & Space & {\small Fundamental Polynomial}
            \\ 
            \hline \rowcolor{mycolor!10}
            Hamming length $n$ & Hamming & $n$ & $\mathbb{F}_q^{n}$ & $X+(q-1)Y$ 
            \\
        \hline 
        \rowcolor{mycolor!20}
            \small{Bilinear $m\times n, ~m\geq n$} &  & $n$ & $\mathbb{F}_{q}^{m\times n}$ &  \\
        \rowcolor{mycolor!20}
            Gabudulin $m\times n$ & \multirow{-2}{2em}{Rank} & $n$ & $\mathbb{F}_{q^m}^{n}$ & \multirow{-2}{4cm}{\centering$X+\left(q^m-1\right)Y$}\\

        \hline \rowcolor{mycolor!30}
             &  & & & $t$ odd, $X+\left(q^{2n+1}-1\right)Y$ \\ 
            \rowcolor{mycolor!30}
            \multirow{-2}{4cm}{\centering Skew $t\times t$, $n=\left\lfloor\frac{t}{2}\right\rfloor$, $m=\frac{t(t-1)}{2n}$} & \multirow{-2}{2cm}{\centering Skew Rank} & \multirow{-2}{1cm}{\centering \small $n$}  & \multirow{-2}{1cm}{\centering $\mathbb{F}_q^{t\times t}$}  & $t$ even, $X+\left(q^{2n-1}-1\right)Y$  \\

        \hline \rowcolor{mycolor!40}
            Hermitian $t\times t$ & Rank & $t$ & $\mathbb{F}_{q^2}^{t\times t}$ & $X+\left(-(-q)^t-1\right)Y$
            \\
        \hline \rowcolor{mycolor!50}
            {\small Krawtchouk A-S. } & - & $n$ & $\mathscr{X}$ & $X+\left(cb^{n}-1\right)Y$
            \\
        \hline
    \end{tabular}
    \caption{An overview of key parameters and results}
    \label{tab:general1}
\end{table}
\begin{table}[H]
    \centering\renewcommand{\arraystretch}{1.4}
    \begin{tabular}{|c|c|c|}
        \hline
            Name & Eigenvalues & Valencies $v_s$
            \\ 
            \hline \rowcolor{mycolor!10}
            Hamming length $n$ & $\sum_{j=0}^k (-1)^j \binom{x}{j}\binom{n-x}{k-j}(q-1)^{k-j}$ & $\binom{n}{s}(q-1)^{s}$
            \\
        \hline 
        \rowcolor{mycolor!20}
            \small{Bilinear $m\times n, ~m\geq n$} &  & \\
        \rowcolor{mycolor!20}
            Gabudulin $m\times n$ & \multirow{-2}{7.2cm}{\centering $\sum_{j=0}^{k}(-1)^{j}q^{j(n-x)}q^{\sigma_{j}} \prescript{}{q}{\bbinom{x}{j}}\prescript{}{q}{\bbinom{n-x}{k-j}}\alpha(n-j,k-j)$} & \multirow{-2}{2cm}{\centering$\prescript{}{q}{\bbinom{n}{s}}\alpha(m,s)$} \\

        \hline \rowcolor{mycolor!30}
             &    & \\ 
            \rowcolor{mycolor!30}
            \multirow{-2}{4cm}{\centering Skew $t\times t$, $n=\left\lfloor\frac{t}{2}\right\rfloor$, $m=\frac{t(t-1)}{2n}$} & \multirow{-2}{7.6cm}{\centering\small $\sum_{j=0}^k (-1)^j q^{2j(n-x)} q^{j(j-1)}\prescript{}{q^2}{\bbinom{x}{j}}\prescript{}{q^2}{\bbinom{n-x}{k-j}}\gamma(m-2j,k-j)$} & \multirow{-2}{2cm}{\centering $\prescript{}{q^2}{\bbinom{n}{s}}\gamma(m,s)$}  \\

        \hline \rowcolor{mycolor!40}
            Hermitian $t\times t$ & {\small $\sum_{j=0}^{k} (-1)^{j}(-q)^{j(t-x)}(-q)^{\sigma_{j}} \prescript{}{-q}{\bbinom{x}{j}}\prescript{}{-q}{\bbinom{t-x}{k-j}}\gamma'(t-j, k-j)$} & $\prescript{}{-q}{\bbinom{t}{s}}\gamma'(t,s)$
            \\
        \hline \rowcolor{mycolor!50}
            {\small Krawtchouk A-S. } & $\sum_{j=0}^k (-1)^j b^{j(n-x)} b^{\sigma_j}\prescript{}{b}{\bbinom{x}{j}}\prescript{}{b}{\bbinom{n-x}{k-j}}\gamma_{b,c}(n-j,k-j)$ & $\prescript{}{b}{\bbinom{n}{s}}\gamma_{b,c}(n,s)$
            \\
        \hline
    \end{tabular}
    \caption{An overview of key parameters and results}
    \label{tab:general3}
\end{table}

To clarify, for the Hermitian association scheme we cannot directly use Delsarte's recurrence relation as $-q$ lies outside the parameters stated. Instead we prove in Proposition \ref{prop:recurrencematch} that the recurrence relation \cite[Lemma 7]{KaiHermitian} used by Schmidt to determine the eigenvalues of this scheme has exactly the same solutions as Delsarte's recurrence relation \cite[(1)]{delsartereccurance} but with the parameters used for the Hermitian association scheme, $b=-q$ and $c=-1$. I 
\begin{prop}\label{prop:recurrencematch}
    For $b=-q\in\mathbb{R}$, $b\neq 0$, $x,k\in\{0,\ldots,n\}$ the recurrence relation from \cite[Lemma 7]{KaiHermitian},
    \begin{equation*}
        C_{k+1}(x+1,n+1) = C_{k+1}(x,n+1) + b^{2n+1-x}C_{k}(x,n)
    \end{equation*}
    has the same solutions as the recurrence relation from \cite[(1)]{delsartereccurance},
    \begin{equation*}
        C_{k+1}(x+1,n+1) = b^{k+1}C_{k+1}(x,n)-b^kC_{k}(x,n)
    \end{equation*}
    where $C_k(x,n)$ are the eigenvalues of the association scheme of Hermitian matrices.
\end{prop}                                    \begin{proof}
    We show that 
    \begin{equation*}
        C_{k+1}(x,n+1) + b^{2n+1-x}C_{k}(x,n) = b^{k+1}C_{k+1}(x,n)-b^kC_{k}(x,n).
    \end{equation*}
    Let $\alpha_1=b^{k+1}C_{k+1}(x,n)$ and $\alpha_2=b^{k}C_{k}(x,n)$ and also let $\beta_1=C_{k+1}(x,n+1)$ and $\beta_2=b^{2n+1-x}C_{k}(x,n)$. Then using the eigenvalues defined in Table \ref{tab:general3} we have,
    \begin{align*}
        \alpha_1 
            & = b^{k+1}\sum_{j=0}^{k+1}(-1)^{j}b^{j(n-x)}b^{\sigma_{j}} \bbinom{x}{j}\bbinom{n-x}{k+1-j}\gamma'(n-j,k+1-j)
        \\
        \alpha_2 
            & = b^{k}\sum_{j=0}^{k}(-1)^{j}b^{j(n-x)}b^{\sigma_{j}} \bbinom{x}{j}\bbinom{n-x}{k-j}\gamma'(n-j,k-j)
        \\
        \beta_1 
            & = \sum_{j=0}^{k+1}(-1)^{j}b^{j(n+1-x)}b^{\sigma_{j}} \bbinom{x}{j}\bbinom{n+1-x}{k+1-j}\gamma'(n+1-j,k+1-j)
        \\
        \beta_2 
            & = b^{2n+1-x}\sum_{j=0}^{k}(-1)^{j}b^{j(n-x)}b^{\sigma_{j}} \bbinom{x}{j}\bbinom{n-x}{k-j}\gamma'(n-j,k-j).
    \end{align*}
    Consider $\alpha_1\vert_{j=k+1}$ and $\beta_1\vert_{j=k+1}$. Then,
    \begin{align*}
        \alpha_1\vert_{j=k+1}
            & = b^{k+1}(-1)^{k+1}b^{(k+1)(n-x)}b^{\sigma_{k+1}}\bbinom{x}{k+1}\bbinom{n-x}{0}\gamma'(n-k+1,0)
            \\
        \beta_1\vert_{j=k+1}
            & = (-1)^{k+1}b^{(k+1)(n+1-x)}b^{\sigma_{k+1}}\bbinom{x}{k+1}\bbinom{n+1-x}{0}\gamma'(n+1-k-1,0)
    \end{align*}
    since $\gamma'(x,0)=1$ for any $x\in\mathbb{R}$. So $\alpha_1\vert_{j=k+1}=\beta_1\vert_{j=k+1}$. Now rearranging $\alpha_1$ and $\beta_1$ we have
    \begin{align*}
        \alpha_1 
            & = b^{k+1}\sum_{j=0}^{k+1}(-1)^{j}b^{j(n-x)}b^{\sigma_{j}} \bbinom{x}{j}\bbinom{n-x}{k+1-j}\gamma'(n-j,k+1-j)
        \\
            \overset{\eqref{equation:gaussianfracxk-1}\eqref{equation:gammastepdownsecond}}&{=}
            \sum_{j=0}^{k} (-1)^{j}b^{j(n-x)+1}b^{\sigma_j}\bbinom{x}{j}\frac{b^{n-x-k+j}-1}{b^{k+1-j}-1}\bbinom{n-x}{k-j}\left(-b^{n-j}-b^{k-j}\right)\gamma'(n-j,k-j)
            \\
            & \hspace{1cm} + \alpha_1\vert_{j=k+1}
        \\
        \beta_1 
            & = \sum_{j=0}^{k+1}(-1)^{j}b^{j(n+1-x)}b^{\sigma_{j}} \bbinom{x}{j}\bbinom{n+1-x}{k+1-j}\gamma'(n+1-j,k+1-j)
        \\
            \overset{\eqref{equation:beta1stepdown}\eqref{equation:gammastepdown}}&{=} \sum_{j=0}^{k}(-1)^{j}b^{j(n+1-x)}b^{\sigma_j}\bbinom{x}{j}\frac{b^{n+1-x}-1}{b^{k+1-j}-1}\bbinom{n-x}{k-j}b^{k-j}\left(-b^{n+1-j}-1\right)\gamma'(n-j,k-j)
            \\
            & \hspace{1cm} + \beta_1\vert_{j=k+1}.
    \end{align*}
    Now let $C=\alpha_1-\alpha_2-\beta_1-\beta_2$. Then
    \begin{align*}
        \alpha_1-\alpha_2-\beta_1-\beta_2
            & = \sum_{j=0}^{k}(-1)^{j}b^{j(n-x)}\beta^{\sigma_j}\bbinom{x}{j}\bbinom{n-x}{k-j}\gamma'(n-j,k-j)
            \\
            & \hspace{1cm} \times \Bigg(b^{k+1}\left(-b^{n-j}-b^{k-j}\right)\frac{\left(b^{n-x-(k-j)}-1\right)}{b^{k+1-j}-1}-b^{k}
            \\
            & \hspace{1cm} -b^{j}b^{k-j}\left(-b^{n+1-j}-1\right)\frac{b^{n+1-x}-1}{b^{k+1-j}-1}-b^{2n+1-x}\Bigg)
            \\
            & \hspace{1cm} +\alpha_1\vert_{j=k+1}-\beta_1\vert_{j=k+1}.
    \end{align*}
    But 
    \begin{align*}
    b^{k+1}\left(-b^{n-j}-b^{k-j}\right)
        & \frac{\left(b^{n-x-(k-j)}-1\right)}{b^{k+1-j}-1}-b^{k} -b^{j}b^{k-j}\left(-b^{n+1-j}-1\right)\frac{\left(b^{n+1-x}-1\right)}{b^{k+1-j}-1}-b^{2n+1-x}
        \\
        & = \frac{b^k}{b^{k+1-k}-1}\Bigg(b^{k+1}\left(-b^{n-j+1}-b^{k-j+1}\right)\left(b^{n-x-k+j}-1\right)
        \\
        & \hspace{1cm} \times \left(-b^{n+1-j}-1\right)\left(b^{n+1-x}-1\right)\Bigg) - b^k -b^{2n+1-x-k}
        \\
        & = \frac{b^k}{b^{k+1-k}-1}\Big(-b^{2n+1-k-x}-b^{n-x+1}+b^{n-j+1}
        \\
        & \hspace{1cm} +b^{k-j+1}+b^{2n+2-x-j}-b^{n+1-j}+b^{n+1-x}-b^{k+1-j}+1
        \\
        & \hspace{1cm} -1-b^{2n+1-x+1-j}+b^{2n+1-x-k}\Big)
        \\
        & = 0.
    \end{align*}
    Therefore $C=0$, and $\alpha_1-\alpha_2=\beta_1+\beta_2$. Therefore the recurrence relations have the same solutions.
\end{proof}     
\newpage
\printbibliography[heading=bibintoc,title={Bibliography}]
\end{document}